\documentclass[12pt]{article}
\usepackage[a4paper, margin=2cm]{geometry} % Sets A4 with 2cm margins
\usepackage{blindtext}
\usepackage{setspace}
\usepackage{amsmath,amsfonts}
\usepackage{algorithmic}
\usepackage{algorithm}
\usepackage{array}
\usepackage[caption=false,font=normalsize,labelfont=sf,textfont=sf]{subfig}
\usepackage{textcomp}
\usepackage{stfloats}
\usepackage{url}
\usepackage{verbatim}
\usepackage{graphicx}
\usepackage{cite}
\usepackage{amsmath}
\usepackage{array}
\usepackage{mathrsfs}
\usepackage{hyperref}
\usepackage{multicol}
\usepackage{amssymb}
\usepackage[linesnumbered,ruled,vlined,algo2e]{algorithm2e}
\usepackage{lineno}
\usepackage{amsthm}
\newtheorem{theorem}{Theorem}
\newtheorem{lemma}{Lemma}

\theoremstyle{definition}

\newtheorem{problem}{Problem}
\newtheorem{corollary}{Corollary}
\newtheorem{discussion}{Discussion}
\newtheorem{assumption}{Assumption}
\newcommand{\cN}{{\mathcal N}}

\newcommand{\cI}{{\mathcal I}}
\newcommand{\cR}{{\mathcal R}}
\newcommand{\cIc}{{\mathcal{I}^c}}
\newcommand{\A}{{\mathcal{A}}}
\newcommand{\W}{{cost}}
\newcommand{\J}{{\mathcal{J}}}

\newcommand{\mO}{{\mathcal{O}}}
\newcommand{\cD}{{\mathcal D}}
\begin{document}

\title{Approximation Algorithm on Drone Delivery Packing Problem with Battery Stations}

\author{Saswata~Jana  and~Partha~Sarathi~Mandal% <-this % stops a space
\thanks{S. Jana is a Research scholar and P. S. Mandal is professor of Indian Institute of Technology Guwahati, Assam, India.}% <-this % stops a space
\thanks{A preliminary version of this paper appeared in Proc. of 24th Int. Conference on Distributed Computing and Networking (ICDCN 2023) \cite{janaPacking}}% <-this % stops a space
\thanks{Manuscript received April 19, 2005; revised August 26, 2015.}}

% The paper headers
% \markboth{Journal of \LaTeX\ Class Files,~Vol.~14, No.~8, August~2021}%
% {Shell \MakeLowercase{\textit{et al.}}: A Sample Article Using IEEEtran.cls for IEEE Journals}

%\IEEEpubid{0000--0000/00\$00.00~\copyright~2021 IEEE}
% Remember, if you use this you must call \IEEEpubidadjcol in the second
% column for its text to clear the IEEEpubid mark.

\maketitle
%\vspace{-5pt}
\begin{abstract}
%Recent advancements in unmanned aerial vehicles (UAVs, or drones) have prompted logistics to adopt drones for multiple operations. 
Collaboration between drones and trucks in a last-mile delivery system offers numerous benefits and reduces many challenges of the traditional delivery system. Here, we introduce \textit{Drone-Delivery Packing Problem}, where a set of parcels, associated with delivery intervals and cost, should be delivered to customer locations. The system comprises a set of identical drones and battery stations along truck’s route, where drones swap depleted batteries or recharge them. The objective is to find assignment for all parcels by using the minimum number of drones, subject to the battery budget and compatibility of each drone's assignment. 
We consider three variants of the problem, based on conflicting characteristics and existence of battery service stations. All are NP-hard, and we have proposed approximation algorithms for each. When there are no battery stations, we propose an approximation algorithm that uses at most $\min\{(\frac{11}{9}OPT_{NC}+\frac{24}{9}), (\frac{3}{2}OPT_{NC} + \frac{3}{2})\}$ drones, where $OPT_{NC}$ is the size of the optimum solution.
When the intervals are non-conflicting, we design a $(2+\psi)$-approximation algorithm. In the presence of both battery stations and conflicting intervals, we present a $(4+\psi)$-approximation algorithm. The algorithm is later modified into a $(3+\psi)$-approximation algorithm when the battery service stations act as swapping stations.
Here $\psi=\frac{\epsilon_{max}-\epsilon_{min}}{1-\epsilon_{max}}$, where $\epsilon_{min}=\frac{1}{B}\min\limits_{1\leq j\leq n}cost(I_j)$, $\epsilon_{max}=\min\{\frac{1}{2},\frac{1}{B}\max\limits_{1\leq j\leq n}cost(I_j)\}$, $cost(I_j)$ is the cost for the delivery $j$, and $B$ is the drone’s battery budget.
Finally, we validate our results and compare the performance with the optimum on different instances.
\end{abstract}

%\begin{keywords}

\section{Introduction}\label{sec:introduction}
\noindent \textit{Motivation:}
The rapid demand for commercial deliveries motivates the e-commerce giants to find more effective ways to deliver parcels to customers.
\textit{Last-mile delivery} \cite{LastMile} is the final step in the delivery journey, where the product must be delivered to the customer's doorstep from the distribution hub.
This stage is often the most expensive and involves significant human interaction.
However, advances in drone technology have enabled miniaturization, leading to contactless delivery systems. Major delivery companies have initiated measures to enable efficient parcel delivery using drones.
%\cite{amazon}, \cite{DHL}.
The integration of ground vehicles (e.g., trucks and vans) with drones, each operating within its respective constraints, improves profitability and reduces overall delivery time. Additionally, delivery packages by drone exhibit a substantial reduction in $CO_2$ emission compared to delivery by ground vehicles \cite{Co2Emissions_Goodchild}.
Furthermore, drones have numerous applications
%in various fields, including a, healthcare, and defence.
agriculture\cite{agriculture}, healthcare 
\cite{healthcare}, defence and disaster response \cite{defence},
emergency response, etc.
% agriculture, healthcare, defence, disaster response, emergency response, etc.

%\vspace{1mm}
\noindent \textit{Challenges:}
For a given set of customer locations, we need to know the truck route, launch and rendezvous point for a delivery. All deliveries must be completed within the specified time windows using a fixed fleet of drones with limited capacity. 
Therefore, the use of charging stations or battery-swapping policies becomes necessary.
%However, battery swapping is significantly more expensive than recharging.
Moreover, a drone cannot be assigned to arbitrary delivery sets due to overlapping time intervals, and at any given time, it can deliver at most one package. These constraints collectively influence the logistics planning, aiming to complete all deliveries with a limited number of identical drones while minimizing the total delivery cost.
%In this paper, we are motivated to minimize the number of drones that we should deploy on the truck beforehand to accomplish all the deliveries.
% We also introduce some charging station on the tuck's path, where the drones can recharge (or swap the fully/paritally drained battery with the fully recharged one) themselves, if needed.
Before presenting our problem and contributions, we briefly review the most relevant problem addressed in the literature \cite{Betti, janaSchedulingJournal} to to motivate further study.

%\vspace{1mm}
\noindent \textit{Drone delivery
scheduling problem:} For a given set of customers and delivery time intervals with the cost, profit, and battery budget of the drones, the goal is to schedule a fixed set of drones for the deliveries to maximize the total profit. This NP-hard problem was formulated by Sorbelli \emph{et al}. \cite{Betti}, who proposed approximation algorithms for both single-drone and multiple-drone cases. Later, Jana and Mandal \cite{janaSchedulingJournal} improved certain approximation factors.

%The aforementioned problem does not guarantee delivery of all deliveries because it uses a fixed set of drones to optimize total profit for delivery.
Since the aforementioned problem uses a fixed fleet of drones, some customers may not receive their deliveries. 
In this paper, we introduce the \textit{Drone-Delivery Packing Problem}, which aims to minimize the number of drones required to complete all deliveries.
Minimizing the number of drones is crucial, even though we assume that the truck carries a sufficient number of drones. In general, a delivery company typically receives requests from multiple zones and must assign several trucks throughout the day. For each zone, a single truck is usually assigned, and the available drones must be distributed across these trucks. Therefore, reducing the number of drones required per zone allows the company to serve more zones using the drones available in its warehouse. 
Additionally, we assume that several battery stations are located along the truck’s route, where drones can recharge or replace their depleted batteries with fully charged ones as needed.
Allowing drones to recharge (or swap batteries) significantly reduces the number of drones required for successful delivery operations.
We define three variants of the problem and then propose separate approximation algorithms for them. 
Two performance ratios are commonly used in the context of approximation algorithm. One is \textit{absolute approximation ratio}, which is $\sup\limits_{\mathcal{I}} \frac{\sf ALG(\mathcal{I})}{\sf OPT(\mathcal{I})}$, where $\sf ALG(\mathcal{I})$ is the size of the solution returned by the algorithm $\sf ALG$ for the instance $\mathcal{I}$ and $\sf OPT(\mathcal{I})$ is the optimal solution size. The other one is \textit{asymptotic approximation ratio}, which is $\limsup \limits_{n \to \infty}\sup\limits_{\mathcal{I}}\{ \frac{\sf ALG(\mathcal{I})}{\sf OPT(\mathcal{I})} ~| ~\sf OPT(\mathcal{I}) = n\}$. Our results, summarized below, include both absolute and asymptotic bounds.
%\vspace{1mm}
\subsection{Contributions}
\label{sec:contribution}
Our contributions in this paper are the following.
\begin{itemize}
    \item We introduce three variants of the Drone Delivery Packing Problem (DDP)(Section~\ref{section:model}). When the input does not contain battery service stations, we refer to the problem as DDP-NS. If the delivery intervals are non-conflicting, we refer to the problem as DDP-NC. In the general case, where battery service stations are present, and delivery intervals may conflict, we refer to the problem as DDP-SC.
   % \item  We prove that all three variants of the problem are NP-hard (Section \ref{section:hardness}).
    \item We present an approximation algorithm for DDP-NS (Section \ref{sec:ddp-ns}), which uses at most  $(2 + \frac{\epsilon_{max}-\epsilon_{min}}{1-\epsilon_{max}})OPT_{NS}$ drones and having running time $\mO(n \log n + n_e)$. Here, $\epsilon_{min} = \frac{1}{B}\min \limits_{1 \leq j \leq n} cost(I_j)$, $\epsilon_{max} = \min\{\frac{1}{2}, \frac{1}{B}\max \limits_{1 \leq j \leq n} cost(I_j)\}$, $cost(I_j)$ is the cost for the delivery $j$, and $B$ is the battery budget of the identical drones. $n$ is the number of deliveries, $n_e$ is the total number of conflicts among the delivery intervals, and $OPT_{NS}$ is the optimum solution size for DDP-NS.
    \item We propose an approximation algorithm for DDP-NC (Section \ref{sec:ddp-nc}) that uses at most $\min\{\frac{11}{9}$ $OPT_{NC} + \frac{24}{9}, \frac{3}{2}OPT_{NC} + \frac{3}{2}\}$ drones with running time $\mathcal{O}(n\log n + n_e + r)$, where $OPT_{NC}$ is the optimum solution for DDP-NC.
   % \item We propose an approximation algorithm for DDP-NC (Section \ref{sec:ddp-nc}) that uses at most $(\frac{11}{9}$ $OPT_{NC} + \frac{24}{9})$ drones with running time $\mathcal{O}(n\log n + n_e + r)$, where $OPT_{NC}$ is the optimum solution for DDP-NC.
  \item  We design an approximation algorithm for DDP-SC (Section \ref{sec:ddp-sc}) that uses at most $(4 + \frac{\epsilon_{max}-\epsilon_{min}}{1-\epsilon_{max}})OPT_{SC}$ drones with running time $\mathcal{O}(n\log n + n_e + r)$, where $OPT_{SC}$ is the optimum solution size for DDP-SC.
  \item We design another approximation algorithm for DDP-SC (Section \ref{sec:mod-ddp-sc}) that uses at most $(3 + \frac{\epsilon_{max}-\epsilon_{min}}{1-\epsilon_{max}})OPT_{SC}$ drones with running time $\mathcal{O}(n^{2.273} + r)$, when the battery service stations are the swapping stations.
  \item We compare our algorithm with the optimum on various instances (Section \ref{sec:experiment}).
\end{itemize}

The additional contribution of this paper, beyond our conference version~\cite{janaPacking}, lies in introducing DDP with battery service stations, which extends the definition of DDP, and in developing approximation algorithms for this variant.

\subsection{Related Work}
Ground vehicles are widely employed in various operations, including the delivery of goods. The problem of serving points of interest (in this case, customer locations) using vehicles is known in the literature as the \textit{routing problem}. Two main variants have been explored in this problem. One is \textit{traveling salesman problem} (TSP), where an uncapacitated vehicle is used to minimize the total time to accomplish the services. The other one is \textit{vehicle routing problem} (VRP), where multiple capacitated vehicles are used to minimize the total makespan.
Several variants \cite{laporte1992traveling},  \cite{anoherworkLaporte} have been studied for these two problems including the time window to serve a customer and locating depots for refilling the capacitated trucks. Additionally, the online version of the vehicle routing problem \cite{jaillet2008online}, where the customer's requests appear online, has also been studied with a similar objective as earlier. VRP with battery recharging and battery swapping has also been studied by the authors in \cite{conrad2011recharging, mao2020electric, verma2018electric}. The results in those papers mainly focus on the heuristics and experimental results.

%Nowadays, use the of drones for delivering goods has become more interesting due to the huge increase in e-commerce sales.
Since drones have limited mobility, integrating trucks enhances delivery efficiency. Extensive research has been conducted on collaborative delivery systems involving both drones and trucks.
This hybrid model of delivery comes into consideration when Murray and Chu \cite{MURRAY201586} introduced the \textit{flying sidekicks travelling salesman problem} (FSTSP), a more extended version of TSP, where customers need to be visited either by truck or by a drone starting from the depot. Meanwhile, the drone begins the journey either from the depot or from any customer location, and the same for the meeting occasion. Here, the authors aim to minimize the total makespan to ensure the completion of all deliveries. They proposed an optimal \textit{mixed integer linear programming} (MILP) formulation and two heuristic solutions.
Crisan and Nechita \cite{CRISAN201938} proposed another effective heuristic for FSTSP by using the solution for TSP. Murray and Raj \cite{Murray2019TheMF} extended FSTSP for multiple drones.
%They proposed MILP formulation for the problem and a heuristic solution with numerical tests.
Daknama and Kraus \cite{Daknama2017VehicleRW} considered the budgeted FSTSP model with the limited flying endurance of drones. A rechargeable area on the truck's roof where the drones can charge after completing each delivery and fly for the next. The authors proposed a heuristic algorithm for scheduling the truck and drones. Wang \emph{et al.} \cite{wang2017vehicle} studied the budgeted drone-delivery model from the worst-case scenario. They proposed a few worst-case factor analyses depending on the number of drones and the relative speed of the drone and the truck. Authors in \cite{ahani2020age}, \cite{kim2018hybrid}, \cite{park2017battery} studied several UAV-based models with recharging policy and proposed heuristics with the final objective of optimizing the total traveling cost, or profit, or energy consumption, etc.
 Delivery by drones only came to the attention in \cite{Boysen2018DroneDF}, where the objective is to find the launch and meet point for delivery among a set of stopping points on the trucks' route such that the total makespan for completing all the deliveries is minimized. The assumption for this problem was the knowledge of the truck's route but without any battery constraint of the drones. 

Sorbelli \emph{et al.} \cite{Betti} proposed a
\textit{multiple drone-delivery scheduling
problem} (MDSP), where a truck and multiple drones cooperate for package delivery in the last mile.
They assumed that the launching and rendezvous points for each delivery were given. 
The authors provided approximation algorithms for the single-drone and multiple-drone cases. Later, Jana and Mandal \cite{janaSchedulingJournal} proposed a few constant-factor approximation algorithms and an exact algorithm to address the MDSP. In the MDSP, some customers may not be served by any drones due to the model's reliance on a fixed set of drones. However, our objective in this paper is to schedule all the deliveries. We have also introduced a battery-recharging policy, with charging stations located along the truck's route. A fully or partially drained drone can recharge itself to resume its delivery journey. Our proposed problem aims to optimize the number of rechargeable drones, assuming the truck's and drone's routes are predefined. Consequently, this problem differs from the well-studied VRP.  
%\vspace{1mm}
% \noindent{\bf Roadmap.} We discuss the
% model and preliminaries, along with the problem definition in Section \ref{section:model}. 
% Next, we present the hardness of the problem along with an ILP formulation in Section \ref{section:hardness}. 
% We propose all the approximation algorithms and provide the analysis in Section \ref{section:aproxAlgos}. 
% Finally, we conclude in Section \ref{section:concl}. Due to space constraints we defered some proofs on the Appendix \ref{appendix}.
\vspace{-5pt}
\section{Model and Problem Definition}
\label{section:model}
%\noindent In this section, we describe the Drone-Delivery Packing Problem with Charging Stations (DDPCS) model with an example. 
%All the variables used in this paper are listed in Table \ref{tab:notation}.
\vspace{-3pt}
\noindent \textbf{Model:} Let $\mathcal{N} = \{1, 2, \ldots, n\}$ denote the set of \textit{deliveries}, where each delivery $j \in \mathcal{N}$ is associated with a customer located at $\delta_j$.
The drones are homogeneous, meaning that each is equipped with an identical battery capacity (or energy budget) denoted by $B$.
Initially, all the drones are at the company's warehouse (depot). A truck containing all the drones departs from the depot and follows a pre-decided path. To complete a delivery $j$ at position $\delta_j$, a drone takes off from the truck at a specific \textit{launching location} $(\delta_j^L)$, and after delivering the package at $\delta_j$, it meets with the truck again at a designated \textit{rendezvous location} $(\delta_j^R)$.
After completing all the deliveries, the truck, carrying all the drones, returns to the depot.
%Along the truck's path, there is a set of battery charging stations where any set of drones can recharge. Let $\mathcal{R} =\{1, 2, \cdots, r\}$ denote the collection of such stations positioned at $\delta_{\ell}^c$ for each $l \in \cR$. 
Let the truck commence its journey from the depot with all the drones at time $t_0$. We denote $t_j^L$ and $t_j^R$ as the times when the truck arrives at the location $\delta_j^L$ and $\delta_j^R$, respectively. We term $t_j^L$ as the \textit{launching time}, $t_j^R$ as the \textit{rendezvous time} and $I_j = [t_j^L, t_j^R]$ as the \textit{delivery time interval} for the delivery $j \in \mathcal{N}$.
$\cI = \{I_1, I_2, \cdots, I_n\}$ be the \textit{delivery time interval set} for the set of deliveries $\mathcal{N}$.
Furthermore, there is a real-valued cost function $cost(\cdot)$ that maps each of the interval $I_j \in \cI$ to $(0, B]$.
The value of $cost(I_j)$ refers to the energy \textit{cost} incurred by a drone to complete delivery $j \in \mathcal{N}$. This cost is measured in terms of the battery consumption by the drone to complete the delivery.
%Let $t_{\ell}^A$ and $t_{\ell}^D$ denote the \textit{arrival} and \textit{departure} times of the truck at charging station $l \in \mathcal{R}$, respectively. We refer to $I_{\ell}^c = [t_{\ell}^A, t_{\ell}^D]$ as the \textit{charging time interval} for the charging station $l \in \mathcal{R}$, during which any set of drones can fully recharge themselves at $l$. 
%A drone, whether fully or partially drained, must spend the entire interval $I_{\ell}^c$ to become fully charged.
%These intervals can also be interpreted as the time required to swap a drained battery with a fully charged one available at the station.
%Let, $\mathcal{I}^c = \{I_1^c, I_2^c, \cdots, I_r^c\}$ be the set of charging time intervals for the set of stations in $\cR$.
%We further assume that no two intervals in $\cR$ are overlapping with each other and they are sorted in $\cR$, i.e., $t_1^A < t_1^D < t_2^A < t_2^D < \cdots < t_r^A < t_r^D$. In addition, if a drone decides to recharge at $l \in \mathcal{R}$, it can not schedule any delivery $j \in \mathcal{N}$ where $I_j \cap I_{\ell}^c \neq \emptyset$.
Also, emphasize that the truck moves uni-directionally along its pre-defined route. So, if $P$ and $Q$ are any two points on the truck's path, where $Q$ occurs later than $P$, then $t_P < t_Q$, where $t_P$ and $t_Q$ are the times when the truck arrives at the positions $P$ and $Q$, respectively.

Any drone can be assigned for multiple deliveries $S$ ($S \subseteq \cI)$ constraints to the battery budget of the drone and the compatibility of the delivery time intervals. Any two delivery time intervals $I_j$ and $I_k$ are said to be \textit{compatible} or \textit{conflict-free} if $I_j \cap I_k = \emptyset$. Otherwise, they are in \textit{conflict}.
%A delivery time interval $I_j$ is \textit{compatible} with a charging time interval $I_{\ell}^c$, if $I_j \cap I_{\ell}^c = \emptyset$.
Two deliveries in $\mathcal{N}$ are said to be in conflict if the respective intervals conflict. Any set of delivery time intervals $S \subseteq \cI$ is said to be \textit{compatible} if all pairs of intervals in it are compatible. A compatible set of deliveries $S \subseteq \cI$ is said to be \emph{feasible} if at any time the addition of a new delivery doesn't violate the battery constraint. A singleton set $\{I_j\}$ is trivially a feasible set and can be assigned to any drone $i$ with battery capacity $B$. After the assignment, the remaining battery capacity of the drone $i$ is decreased by $cost(I_j)$ and equals $(B - cost(I_j))$.
We use $S_i^t$ to denote the feasible set of deliveries and $rem_i^t$ to denote the remaining battery capacity of the drone $i$ at time $t$.
For $I_j \notin S_i^t$, the set $S_i^t \cup \{I_j\}$ is feasible if $S_i^t \cup \{I_j\}$ is compatible and $cost(I_j) \leq rem_i^t$. We define $\W(S) = \sum_{I_j \in S} cost(I_j)$ as the energy cost for the set $S \subseteq \cI$.

%A drone $i$ may recharge at station $l$ between two consecutive delivery intervals $I_j$ and $I_k$ (where $t_j^L < t_k^L$), provided the compatibility constraint $I_j \cap I_{\ell}^c = I_k \cap I_{\ell}^c = \emptyset$ holds.
%In that case, after the departure time $t_{\ell}^D$ from the station $l \in \cR$, the remaining battery of the drone $i$ is set to $B$ again. So an \textit{feasible assignment} of a drone $i$ is the union of some delivery time intervals and charging time intervals, where each charging interval separates two feasible sets of deliveries. We formally define \textit{feasible assignment} of a drone along with the problem statement below. A drone is called \textit{used} if at least one delivery is assigned to it. Note that the superscript $^c$ is used with different variables to distinguish the context of charging stations from that of deliveries. We will explicitly use the indices $i$, $j$, and $l$ to mean a drone, a delivery, and a charging station, respectively.
%We now discuss about two types of battery charging policy in the following.
In addition to the deliveries, there is a set of \emph{battery service stations} along the truck route where a drone can either recharge depleted batteries or swap them with the fully charged ones. Let $\mathcal{R} = \{1, 2, \ldots, r\}$
denote the set of such stations, where station $\ell \in \mathcal{R}$ is located at position $\delta_{\ell}^c$.
Let $t_{\ell}^A$ and $t_{\ell}^D$ be the \emph{arrival} and \emph{departure} times of the truck at station $\ell \in \mathcal{R}$, respectively.
We denote the interval $[t_{\ell}^A, t_{\ell}^D]$ as $I_{\ell}^c$ and referred to as the \textit{waiting time interval} for the truck at the station $\ell \in \mathcal{R}$.
We use $\cI^c$ to represent the set of all waiting time intervals $\{\cI_1^c, \cI_2^c, \cdots, \cI_r^c\}$ for the set of stations $\cR$.
We assume that no two intervals in $\cI^c$ are overlapping with each other and they are sorted in $\cR$, i.e., $t_1^A < t_1^D < t_2^A < t_2^D < \cdots < t_r^A < t_r^D$.
We further assume that all waiting time intervals are sufficiently short compared to the delivery time intervals, so that no delivery interval $I_j$ is completely contained within a waiting time interval. If such a delivery interval exists, it must be located very close to the truck route. Therefore, we neglect these intervals for drone delivery and assume that they can be served directly by the truck.
Moreover, we assume that no delivery interval intersects multiple waiting intervals. This assumption is reasonable because battery service stations are typically located far apart. If a delivery interval overlaps two waiting intervals, the corresponding delivery location would likely be far from the truck route, making the delivery cost extremely high and impractical for drone delivery.
%A drone $i$ may recharge at station $l$ between two consecutive delivery intervals $I_j$ and $I_k$ (where $t_j^L < t_k^L$), provided the compatibility constraint $I_j \cap I_{\ell}^c = I_k \cap I_{\ell}^c = \emptyset$ holds.
%In that case, after the departure time $t_{\ell}^D$ from the station $l \in \cR$, the remaining battery of the drone $i$ is set to $B$ again.

We consider two variants of battery service stations. If a station is used for swapping partially or fully depleted drone batteries with fully charged ones, we refer to it as a \emph{swapping station}. On the other hand, if a station is used for recharging batteries, we refer to it as a \emph{charging station}. The detailed assumptions on the battery stations are described below.

\noindent\emph{Assumptions on Swapping Policy:}
For a swapping station $\ell \in \mathcal{R}$, we explicitly call the waiting time interval $I_{\ell}^c = [t_{\ell}^A, t_{\ell}^D]$ as the \emph{swapping time interval}. In this case, if a drone chooses to swap its depleted battery with the fully charged ones at the station $\ell \in \cR$, it must spend the entire interval $I_{\ell}^c$ at the station. This is because the battery swapping operation requires a fixed amount of time, regardless of the remaining capacity of the drone. Since different stations may have different efficiencies, the lengths of the swapping time intervals may vary across stations.
Furthermore, if a drone chooses to swap its battery at station $\ell \in \mathcal{R}$, due to compatibility, it cannot be scheduled for any delivery $j \in \mathcal{N}$ whose delivery interval $I_j$ intersects with $I_{\ell}^c$.% i.e., $I_j \cap I_{\ell}^c \neq \emptyset$.

\noindent\emph{Assumptions on Recharging Policy:}
In this model, a drone that decides to recharge at a charging station $\ell \in \mathcal{R}$ is allowed to select a sub-interval $I_{\ell}^{c'} = [t_{\ell}', t_{\ell}''] \subseteq I_{\ell}^c$, which we refer to as the \textit{charging time interval} for that drone at the station $\ell$. Two drones can select different charging time intervals at the same station $\ell$. The amount of battery capacity regained by the drone depends on the length of $I_{\ell}^{c'}$.
We assume the existence of a function $Ch_{\ell}(\cdot)$ that takes input as the chosen interval $I_{\ell}^{c'}$ and the remaining battery capacity $rem^{t_{\ell}'}_i$ of a drone $i$ at the time $t_{\ell}'$. Then outputs $rem^{t_{\ell}''}_i$, the battery capacity achieved after recharging (after the time $t_{\ell}''$) at station $\ell$. We assume that, for any value of $rem^{t_{\ell}'}_i$,
$Ch_{\ell}(I_{\ell}^c, rem^{t_{\ell}'}_i) = B$, i.e., if a drone spend entire duration of the waiting time of the truck at the station $\ell$ for its recharging, then the drone will be fully charged. Furthermore, for any sub-interval $I_{\ell}^{c'} \subseteq I_{\ell}^c$ and for any value of $rem^{t_{\ell}'}_i$,~ $Ch_{\ell}(I_{\ell}^{c'}, rem^{t_{\ell}'}_i) \leq B$, i.e., the post battery level after recharging can not exceed the maximum attainable capacity $B$. In addition, if a drone chooses the interval $I_{\ell}^{c'}$ for charging, then it can not be scheduled for any delivery $j$ whose delivery interval $I_j$ intersects with $I_{\ell}^{c'}$.

At any station, multiple drones can simultaneously swap or recharge their batteries. A \emph{feasible assignment} of a drone consists of both the delivery intervals as well as swapping (charging) intervals subject to the \emph{compatibility} (no two intervals are in conflict) and \emph{budget constraints} (at any time, new inclusion of a delivery does not exceed the remaining capacity of the drone). 

In this paper, we consider three variants of the problem, depending on the compatibility characteristics of the given delivery time interval set $(\cI)$ and the existence of battery service stations $(\cR)$ along the truck route. If at least one pair of intervals in $\cI$ is in conflict, then we call the set $\cI$ \emph{conflicting}; otherwise, \emph{non-conflicting}. If there are no battery service stations along the truck route, i.e., $\cR = \emptyset$, we refer to this problem as the \emph{Drone Delivery Packing Problem without Battery Stations}, or, in short, DDP-NS. On the other hand, if the given delivery time interval set is non-conflicting, we refer to this problem as the \emph{Drone Delivery Packing Problem without Conflicting Intervals}, or, in short, DDP-NC. Whereas, in general, if the given delivery interval set is conflicting or non-conflicting and $\cR \ne \emptyset$, we call the problem \emph{Drone Delivery Packing Problem with Battery Stations and Conflicting Intervals}, or, in short, DDP-SC.
We formally define all three variants of the problem below.
A drone is called \textit{used} if at least one delivery is assigned to it. Note that the superscript $^c$ is used with different variables to distinguish the context of battery stations from that of deliveries. We explicitly use the indices $i$, $j$, and $\ell$ to mean a drone, a delivery, and a battery station, respectively. All the variables used in this paper are listed in Table \ref{tab:notation}. %(Section \ref{sec:appendix}).

\begin{problem}
\textbf{Drone Delivery Packing Problem without Battery Stations (DDP-NS).}
\label{def:DDP-NS}
Given a set of delivery time intervals $\cI = \{I_1, I_2, \cdots, I_n\}$ corresponding to the set of deliveries $\mathcal{N}$, and a cost function $cost(\cdot)$ defined on $\cI$, the objective for DDP-NS is to find the minimum number of drones with identical battery capacity such that each delivery is completed by exactly one drone and the assignment of each used drone is feasible.
\end{problem}

\begin{problem}
\textbf{Drone Delivery Packing Problem without Conflicting Intervals (DDP-NC).}
\label{def:DDP-NC}
Given a set of non-conflicting delivery time intervals $\cI = \{I_1, I_2, \cdots, I_n\}$ corresponding to a set of deliveries $\mathcal{N}$, a cost function $cost(\cdot)$ defined on $\cI$, and a set of waiting time intervals $\cI^c =\{I_1^c, I_2^c, \cdots, I_r^c\}$ corresponding to a set of battery service stations $\cR$, the objective for DDP-NC is to minimize number of drones with identical battery capacity such that each delivery is completed by exactly one drone and the assignments of each used drone is feasible.
\end{problem}

\begin{problem}
\textbf{Drone Delivery Packing Problem with Battery Stations and Conflicting Intervals (DDP-SC).}
\label{def:DDP-SC}
Given a set of delivery time intervals $\cI = \{I_1, I_2, \cdots, I_n\}$ corresponding to the set of deliveries $\mathcal{N}$, a cost function $cost(\cdot)$ defined on $\cI$, and a set of waiting time intervals $\cI^c =\{I_1^c, I_2^c, \cdots, I_r^c\}$ corresponding to a set of battery service stations $\cR$, the objective for DDP-SC is to find the minimum number of drones with identical battery capacity such that each delivery is completed by exactly one drone and the assignment of each used drone is feasible.
\end{problem}

%In other words, find the smallest cardinality set of drones $\cM^* = \{1, 2, \cdots, m^*\}$ along with a family of assignments $\{\A_1^*, \A_2^*, \cdots, \A_{m^*}^*\}$, where $\A^*_i \subseteq \cI \cup \cIc$ is a feasible assignment for the drone $i \in \cM^*$, and each delivery interval in $\cI$ is associated with a unique $\A_i^*$,  $\forall 1 \leq i \leq m^*$.

We denote $OPT_{NS}$, $OPT_{NC}$, and $OPT_{SC}$ to refer to the minimum number of drones needed for DDP-NS, DDP-NC, and DDP-SC, respectively.
%Let, $\epsilon_{min} = \frac{1}{B} \min_{j \in \mathcal{N}} cost(I_j)$ $(0 < \epsilon_{min} \leq 1)$.
%We assume that $\epsilon_{min} \leq \frac{1}{2}$. Otherwise, every delivery has cost more than $\frac{B}{2}$. This implies we always need a new drone or a fully recharged used drone to complete all the deliveries, which is the optimum solution. So, we treat this as a trivial case. 
%Also note that if the maximum clique size of the interval graph $G$, is $\omega$, then the optimum solution size of DDPCS (or DDP) is at least $\omega$.
Let $\omega$ denote the maximum number of pairwise conflicting intervals in $\mathcal{I}$, i.e., the size of (in cardinality) the largest mutually conflicting subset of intervals in $\mathcal{I}$. Then the optimal solution size for any of the problems is at least $\omega$.
This is because there must be a time $t$ that is contained in $\omega$ distinct intervals, and to accomplish those deliveries, we need $\omega$ many distinct drones. Thus, 
\vspace{-3pt}
\begin{equation}
\label{eq:rel-w-opt}
    \omega \leq OPT_{NS};~ \omega \leq OPT_{NC};~ \omega \le OPT_{SC}
\end{equation}
\vspace{-3pt}
We now describe the proposed model using the example specified below.
%\vspace{-8mm}
% \begin{figure}[h]
%     \centering  \includegraphics[width=6cm]{intervals_recharge.pdf}
%     \caption{Delivery time intervals $\cI = \{I_1, I_2, \cdots, I_8\}$ where 
%     $I_j = [t_j^L, t_j^R]$ for $1 \leq j \leq 8$ and\\ charging time intervals $\cIc = \{I_1^c, I_2^c\}$ where $I_{\ell}^c = [t_{\ell}^A, t_{\ell}^D]$ for $1 \leq l\leq 2$.}
%     \label{Delivery_Interval}
% \end{figure}

%\begin{figure}%[!ht]
%\begin{minipage}{0.59\textwidth}

%\usepackage{subfig}

\begin{figure}[htbp]
\centering
\begin{minipage}{0.5\columnwidth}
    \centering
    \includegraphics[width=0.8\linewidth]{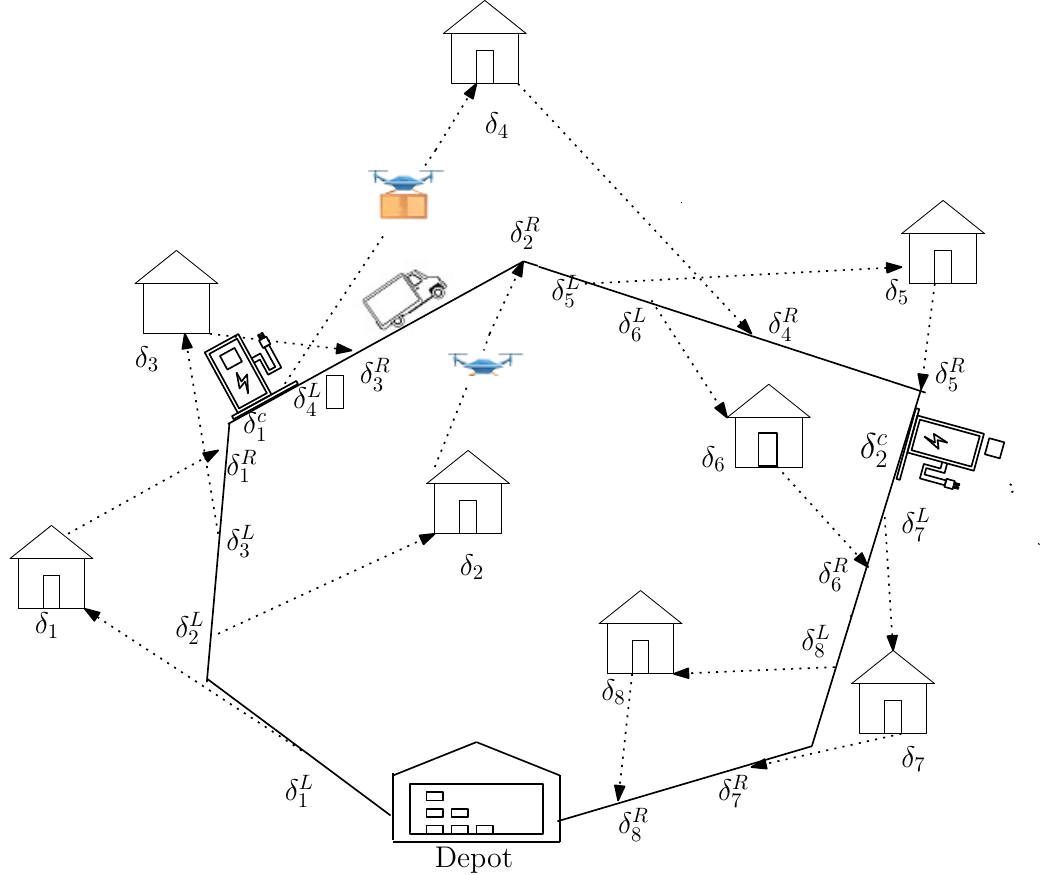}
    \caption{A drone-delivery model.}
    \label{Model}
\end{minipage}
\hfill
\begin{minipage}{0.47\columnwidth}
    \centering
    \includegraphics[width=0.8\linewidth]{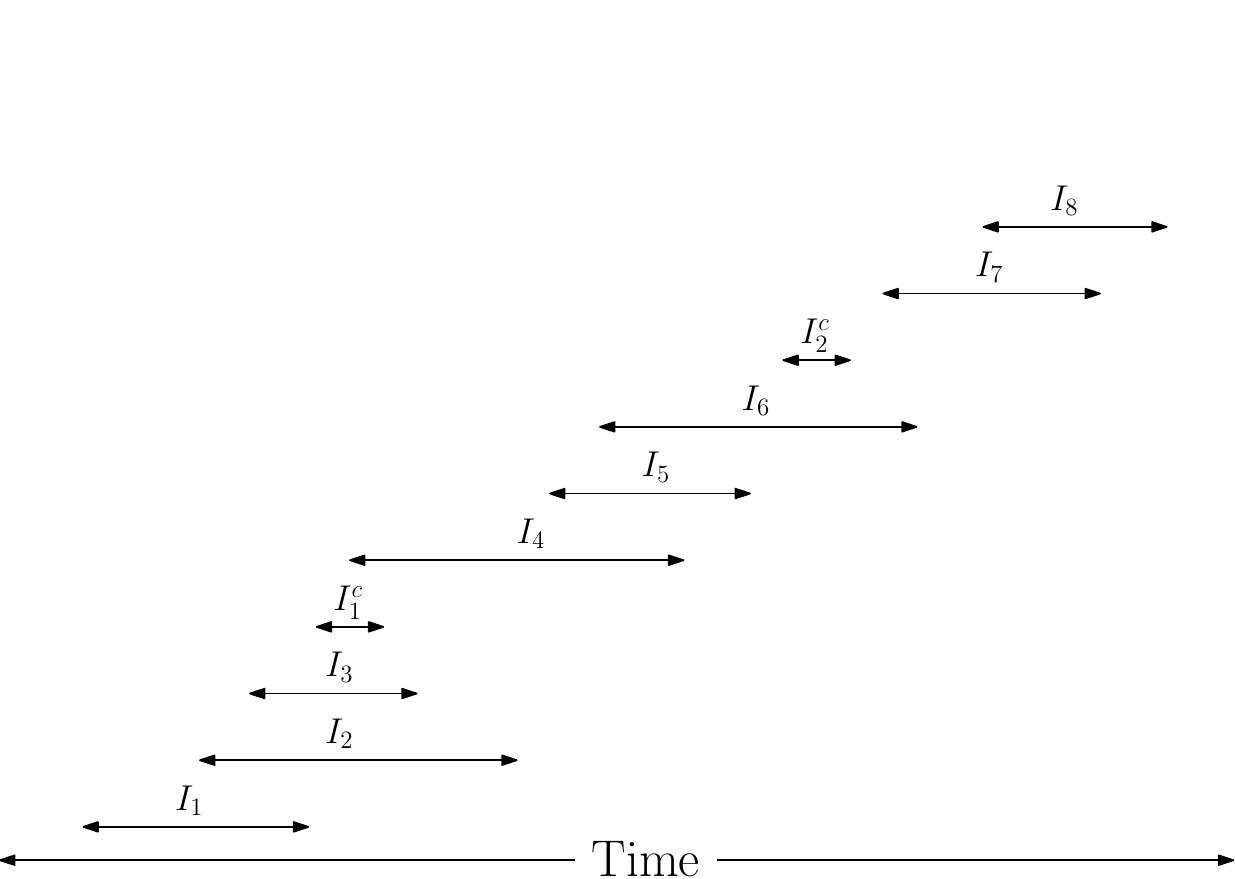}
    \caption{Representation of delivery time intervals $\cI = \{I_1, I_2, \cdots, I_8\}$ and swapping time intervals $\cIc = \{I_1^c, I_2^c\}$}
     \label{Delivery_Interval}
\end{minipage}
\end{figure}

% \begin{figure}%[h]
%     \centering
%     \includegraphics[width=5.9cm]{drone_path_recharge.pdf}
%     \caption{A drone-delivery model.}
%     \label{Model}
% \end{figure}
   
% %\end{minipage} \hfill
% %\begin{minipage}{0.39\textwidth}
% \begin{figure}%[H]
%     \centering  \includegraphics[width=5.33cm]{intervals_recharge.pdf}
%    % \caption{Representation of delivery time intervals $\cI = \{I_1, I_2, \cdots, I_8\}$ where 
%     % $I_j = [t_j^L, t_j^R]$ for $1 \leq j \leq 8$ and waiting time intervals $\cIc = \{I_1^c, I_2^c\}$ where $I_{\ell}^c = [t_{\ell}^A, t_{\ell}^D]$ for $1 \leq l\leq 2$.}
%     \caption{Representation of delivery time intervals $\cI = \{I_1, I_2, \cdots, I_8\}$ and swapping time intervals $\cIc = \{I_1^c, I_2^c\}$}
%     \label{Delivery_Interval}
%     \end{figure}
% \end{minipage}
% \end{figure}

% \begin{figure}[h]
% \centering
% \begin{tikzpicture}
% [node distance={15mm}, thick,main/.style = {draw, circle}] 
% \node[main] (1) {$I_1$};
% \node[main] (2) [right of=1] {$I_2$};
% \node[main] (3) [right of=2] {$I_3$};
% \node[main] (4) [right of=3] {$I_4$};
% \node[main] (5) [below of=4] {$I_5$};
% \node[main] (6) [left of=5] {$I_6$};
% \node[main] (7) [left of=6] {$I_7$};
% \node[main] (8) [left of=7] {$I_8$};
% \draw (1) -- (2);
% \draw (2) -- (3);
% \draw (3) -- (4);
% \draw (5) -- (4);
% \draw (5) -- (6);
% \draw (6) -- (7);
% \draw (7) -- (8);
% \draw (1) to [out=60,in=120,looseness=.5] (3);
% \draw (2) to [out=60,in=120,looseness=.5] (4);
% \draw (4) -- (6);
% \end{tikzpicture} 
% \caption{Interval Graph, $G=(V,E)$.}
% \label{Interval_Graph}
% \end{figure}

%\vspace{2mm}
\noindent{\bf Example.} Figure \ref{Model} shows an example of a drone-delivery model with eight delivery locations, two battery swapping stations. The solid lines in the figure represent the paths of the truck, while the dotted lines represent the paths of the drone. Figure \ref{Delivery_Interval} illustrates the delivery and swapping time intervals associated with the corresponding customers and battery stations shown in Figure \ref{Model}. 
Let $cost(I_1) = 6,~ cost(I_2) = 8,~ cost(I_3)= 4,~ cost(I_4) = 9,~ cost(I_5) = 5,~ cost(I_6) = 7,~ cost(I_7) = 5,~ cost(I_8) = 6$ and $B  = 10$. Then an optimal solution of DDP-SC for this instance uses $4$ drones with the family of assignments $\{\A_1, \A_2, \A_3, \A_4\}$, where $\A_1 = \{I_1, I_1^c, I_6\}, \A_2 = \{I_2\}, \A_3 = \{I_3, I_5, I_2^c, I_7\}$ and $\A_4=\{I_4, I_2^c, I_8\}$.
However, if the above instance does not include any battery station, then it becomes an instance of DDP-NS. An optimal solution for this instance of DDP-NS requires $6$ drones. One drone is scheduled to serve $I_3$ and $I_8$, and another drone is assigned to $I_5$ and $I_7$. Each remaining drone is assigned exactly one delivery time interval. Note that, for this particular instance, $\omega$, the size of the largest set of pairwise conflicting intervals, is three.

\noindent \textbf{Problem Hardness.} All three problems, DDP-NS, DDP-NC, and DDP-SC, are NP-hard, as these are the generalization of the classical \textit{bin packing problem} (BPP) \cite{Coffman1984}. To relate our problem to DDP, we can associate the given interval set with a set of items in BPP, and drones to bins. NP-hardness of the problems motivates us to design some approximation algorithms. We propose separate approximation algorithms for each of the problems in the subsequent sections. We start with an algorithm for DDP-NS, where there are no battery service stations along the truck route. We use `delivery', `delivery interval' and `delivery time interval' interchangeably as synonyms in the rest of the paper.

\section{Approximation Algorithm for DDP-NS}
\label{sec:ddp-ns}Here we demonstrates \textsc{AlgoFor-DDP-NS} (Algorithm \ref{Alg:approx_ddp_color}) using coloring to solve the DDP-NS. The algorithm starts with the construction of an an interval graph $G = (V, E)$ from the given set of delivery intervals $\cI$.
The vertices of $G$ represent the intervals in $\cI$, and two vertices are adjacent if the corresponding two intervals are in conflict. Therefore, $|V| = |\cI| = n$ and we define $n_e = |E|$, total number of conflicts among the intervals in $\cI$.
This construction inherently produces the interval graph, and being so $\chi(G)$ = $\omega(G)$ \cite{west}, where $\chi(G)$ is the \textit{chromatic number} and $\omega(G)$ is the \textit{maximum clique size} of $G$. Moreover, the vertices of $G$ can be colored by $\omega(G)$ many colors in polynomial time such that adjacent vertices get the different color.
We write $\omega$ in place of $\omega(G)$ henceforth, as this is equivalent to the size of the largest set of pairwise conflicting intervals in $\cI$.
From one of the $\omega$-coloring of $G$, we can partition the set $\cI$ into $\omega$ many sets, each corresponding to the same colored vertices of $G$. 
Let $\{\J_1, \J_2, \cdots, \J_{\omega}\}$ be the partition set, where $\J_k \subseteq \cI$ is a compatible (non-conflicting) set of intervals associated with color $k$ $(1 \leq k \leq \omega)$.

For each $k$ $(1 \le k \le \omega)$, we determine the number of drones required to schedule all intervals in $\mathcal{J}_k$, along with their corresponding schedules, using \textsc{GreedyAlgo} (Algorithm~\ref{GreedyAlgo_mod}). The algorithm follows a greedy approach similar to the algorithm used for \textit{bin packing problem} \cite{Coffman1984}. 
The algorithm selects any of the intervals and assigns it to one of the drones.
%Firstly, it sorts the deliveries based on their launch times. Then, it selects the delivery with the least launch time and assigns it to one of the drones.
The assigned drone is referred to as \textit{used} drone. Then it takes the remaining deliveries one by one and attempts to assign them to a previously used drone, subject to battery constraints. 
A new drone is introduced if any delivery does not fit into any of the drones previously used for the interval set $\mathcal{J}_k$.

We use a balanced binary search tree for efficiently implementing the greedy approach, where each \emph{node} in the tree corresponds to an \emph{used drone}.
The node's key $(node.key)$ is the drone's remaining battery capacity. Furthermore, we also store the index, and the set of assigned deliveries to the drone in the attributes, $node.index$, and $node.list$, respectively. We also use some operations in our algorithm.
The operation \textsc{ListInsert}($list$, $I_j$) is used to insert the delivery $I_j$ into the head of $list$.
Whereas \textsc{TreeUpdate}($root$, $node$, $rem$) is used to update $node.key$ by $rem$.
\textsc{TreeInsert}($root$, $rem$, $i$, $I_j$) is employed to insert a new $node$ with $node.key = rem$, $node.index = i$, and create a new list containing only one delivery $I_j$ and assign to $node.list$.
With the help of this tree, we can efficiently find either the index of an open drone to which we can assign a particular delivery (in addition to the existing assignment) or obtain confirmation to introduce a new drone. The pseudocode for this greedy approach, along with the data structure, is depicted in Algorithm \ref{GreedyAlgo_mod}. The detailed description implicitly establishes the correctness of the algorithm.

\begin{figure}[htbp]
\centering
\linenumbers
\begin{minipage}{0.51\columnwidth}
\begin{algorithm}[H]
\LinesNumbered
%\linenumbers
\footnotesize{
\textbf{Initialize:} $m_k = 0$, $root_k = \text{NULL}$\\
\For{each $I_j$ in $\J_k$}{
    $node \gets \textsc{Find}(root_k, cost(I_j))$\\
    \If{$node \neq \text{NULL}$}{
        \textsc{TreeUpdate}($root_k, node, node.key - cost(I_j)$)\\
        \textsc{ListInsert}($node.list, I_j$)\\
    }
    \Else{
        $m_k \gets m_k + 1$\\
        \textsc{TreeInsert}($root_k, B - cost(I_j), m_k, I_j$)\\
    }
}

\textbf{return} $\langle m_k, root_k \rangle$\;

\caption{\textsc{GreedyAlgo}($\J_k, cost(\cdot), B$)}}
\label{GreedyAlgo_mod}

\end{algorithm}
\end{minipage}
\hfill
\begin{minipage}{0.45\columnwidth}

\begin{algorithm}[H]
\caption{\textsc{Find}($node, cost(I_j)$)}
\footnotesize{
\If{$node = \text{NULL}$}{
    \textbf{return} $\text{NULL}$\;
}
\Else{
    \uIf{$node.key \geq cost(I_j)$}{
        \textbf{return} $node$\;
    }
    \ElseIf{$node.right \neq \text{NULL}$}{
        \textbf{return} \textsc{Find}($node.right, cost(I_j)$)\;
    }
    \Else{
        \textbf{return} $\text{NULL}$\;
    }
}
}
\label{Find_mod}
\end{algorithm}

\end{minipage}
\end{figure}

\noindent \textit{Description of Algorithm} \ref{GreedyAlgo_mod} : The inputs of the algorithm are a set of compatible delivery intervals $\mathcal{J}_k (1\le k \le \omega)$, the cost function $cost(\cdot)$, and the battery budget $B$ of the drone.
For assigning a delivery $I_j$ in $\J_k$, 
%according to their non-decreasing launching time,
algorithm calls \textsc{Find}($root_k, cost(I_j)$).
The subroutine \textsc{Find}() (Algorithm \ref{Find_mod}) finds the feasible $node$ in the tree that has remaining capacity at least $cost(I_j)$. If $root.key \geq cost(I_j)$, then the algorithm returns the $root$. Otherwise, the algorithm recurses on its right sub-tree. Note that, when $root.key > cost(I_j)$, the feasible node for the delivery $I_j$ may exists only on the right sub-tree of $root$, as all the $node$ in the left sub-tree of $root$ having $node.key \leq root.key < cost(I_j)$.
In case of the current node having $node.key < cost(I_j)$ with an empty right sub-tree, \textsc{Find()} returns $NULL$.

If \textsc{Find}($root_k, cost(I_j)$) returns $NULL$, a new drone is introduced with index $(m_k+1)$ (initially, $m_k = 0$) and assign the delivery $I_j$ there. Subsequently, a new $node$ is inserted in the tree.
%with $node.key = (B - cost(I_j))$; $node.index = (m_k+1)$ and create a list containing single interval $I_j$ and then assigned to $node.list$ (Line 10).
If \textsc{Find}($root_k, cost(I_j)$) returns a non-null pointer of a $node$ then the algorithm assigns the delivery $I_j$ at the $node$ and update the $node$ suitably.
%by decreasing $node.key$ by $cost(I_j)$ (Line 6). Additionally, we add $I_j$ in the list $node.list$ (Line 7).
At the end, algorithm returns $m_k$, the number of used drones for the interval set $\mathcal{J}_k$, along with the balanced binary search tree rooted at $root_k$, where each node stores the list of deliveries from $\mathcal{J}_k$.

Moreover, we use $k_i$ to denote the $i$-th used drone for the intervals set $\J_k$, and the corresponding assignment is denoted by $S_{k_i}$, where $1 \le k \le \omega$.
Let $\epsilon_k = \frac{1}{B}\max \limits_{I_j \in \J_k} ~cost(I_j)$ and $\epsilon'_k = \min\{\frac{1}{2}, \epsilon_k\}$.
We define $\epsilon_{min} = \frac{1}{B} \min \limits_{I_j \in \mathcal{J}_k}~cost(I_j)$ and $\epsilon_{max} = \max \limits_{1 \leq k \leq \omega} \epsilon'_k$. Then $\epsilon_{max} = \min\{\frac{1}{2}, \frac{1}{B}\max \limits_{j \in \cN}~cost(I_j)\}$. Now, we can state the following lemmas.

\begin{algorithm2e}
\DontPrintSemicolon
\textbf{Input:}$\cI = \{I_1, I_2, \cdots, I_n\}$; $cost(\cdot)$; and $B$.\;
 Construct an interval graph $G$ from the delivery time interval set $\cI$\;
 Find maximum clique size $(\omega)$ of $G$.\;
 Color all the vertices of $G$ with the colors $\{1, 2, \cdots, \omega\}$ such that no two adjacent vertices get the same color.\;
 Find $\J_k$ = \{Set of intervals in $\cI$ whose corresponding vertices in $G$ are colored with the color $k$\} $(1 \leq k \leq \omega)$.\;
 For each $\J_k$ $(1 \leq k \leq \omega)$, find number of drones, say $m_k$ and corresponding assignments, say $\mathcal{S}_k$ by using the Algorithm \ref{GreedyAlgo_mod}.\;
 Return $\sum_{k=1}^{\omega} m_k$ along with their corresponding assignments.\;
 \caption{\textsc{AlgoFor-DDP-NS}}\label{Alg:approx_ddp_color}
\end{algorithm2e}

\begin{lemma}
\label{lemma-color}
$cost(\J_k) \geq  ((m_k - 1) (B - \epsilon'_k B)) + \epsilon_{min} B$, where $m_k$ is the number of drones returned by Algorithm \ref{GreedyAlgo_mod} for the interval set $\J_k \subseteq \cI$, $\epsilon_{min} = \frac{1}{B} \min_{j \in \mathcal{N}} cost(I_j)$ and $cost(\J_k) = \sum \limits_{I_j \in \J_k} ~cost(I_j)$.
\end{lemma}
\begin{proof}
Let $I_{j}$ be the delivery for which the last drone $k_{m_k}$ (the $m_k$-th drone) was used for the set of deliveries $\mathcal{J}_k$. Then each of the used drones, except the drone $k_{m_k}$, has used battery capacity at least $(B - \epsilon_k B)$, as $cost(I_j) \leq \epsilon_k B$. Whereas drone $k_{m_k}$ has used battery capacity at least $\epsilon_{min} B$. So,
\vspace{-3pt}
\begin{equation}
\label{eq:10}
cost(\mathcal{J}_k) \geq  (m_k - 1) (B - \epsilon_k B) + \epsilon_{min} B.
\end{equation}
We also have, $cost(S_{k_i}) \geq \frac{B}{2}$ for $( 1 \leq i \leq m_k-1)$. Otherwise, any two assignments can be merged into a single one. Therefore, 
\vspace{-3pt}
\begin{equation}
\label{eq:11}
cost(\mathcal{J}_k) \geq  (m_k - 1) \frac{B}{2} + \epsilon_{min} B.
\end{equation}
\vspace{-5pt}
Thus from (\ref{eq:10}) and (\ref{eq:11}), we have
$cost(\mathcal{J}_k) \geq  ((m_k - 1) (B - \epsilon'_k B)) + \epsilon_{min} B$, \text{as} $\epsilon'_k = \min\{\frac{1}{2}, \epsilon_k\}$, \text{implies} $(1-\epsilon'_k) = \max\{\frac{1}{2}, 1 - \epsilon_k\}$.%\\
\end{proof}
\begin{corollary}
\label{cor:weight}
    $cost(\mathcal{J}_k) \geq  (m_k - 1) (1 - \epsilon_{max})B + \epsilon_{min} B$, where $m_k$ is the number of drones returned by Algorithm \ref{GreedyAlgo_mod} for the interval set $\J_k \subseteq \cI$, where $\epsilon_{max} = \max \limits_{1 \le k \le \omega} \epsilon_{k'} = \min\{\frac{1}{2}, \frac{1}{B}\max \limits_{j \in \cN}~cost(I_j)\}$ and $\epsilon_{min} = \frac{1}{B} \min \limits_{j \in \mathcal{N}} ~cost(I_j)$.
\end{corollary}

% \begin{corollary}
%     \label{cor:cost-lowerBound}
%     $cost(\mathcal{I}) > (m-1)\frac{B}{2}$. This holds from the Eq. \ref{eq:10} and summing over all $k: 1\le k \le \omega$.
% \end{corollary}

\begin{lemma}
\label{lem:cor+compl-ddp-ns}
Algorithm \ref{Alg:approx_ddp_color} returns a feasible set of assignments in $\mathcal{O}(n\log n + n_e)$ time, where $n$ is the number of deliveries and $n_e$ is the total number of conflicts among the intervals.% in $\cI$.
\end{lemma}
\begin{proof}
    An assignment returned by the Algorithm \ref{Alg:approx_ddp_color} consists of the intervals with the same color. So, they are compatible with each other. Now, to assign a delivery $I_j$ to a drone, the algorithm finds a $node$ in the balanced binary search tree with $node.key \leq cost(I_j)$. If such $node$ is found, the inclusion of $I_j$ to $node.list$, the list of intervals previously assigned to the $node$, keeps the budget constraint satisfied.
    Consequently, all the deliveries in the $node.list$ constitutes a feasible assignment, and can be delivered by a single drone. 

    \noindent \textit{Time Complexity:} We can construct an interval graph $G$ (line 2) from the given interval set $\cI$ in $\mathcal{O}(n + n_e)$ time. Then, finding the maximum clique size ($\omega$) of the graph (line 3) and coloring all the vertices of $G$ (line 4) using $\omega$ colors can be accomplished in $\mathcal{O}(n + n_e)$ time \cite{chordal}. Finding all the $\J_k$ $(1 \leq k \leq \omega)$  (line $5$) takes $\mathcal{O}(n)$ time.
    For each delivery time interval set $\J_k$ $(1 \leq k \leq \omega)$, Algorithm \ref{Alg:approx_ddp_color} uses Algorithm \ref{GreedyAlgo_mod} for finding the number of drones and corresponding assignments  (line $6$). 
To assign a delivery $I_j$ in $\J_k$, Algorithm \ref{GreedyAlgo_mod} calls the subroutine \textsc{Find}($root, cost(I_j)$. This returns a pointer to a node or $NULL$ in $\mathcal{O}(h_j^k)$ time, as the subroutine recurses only on one sub-tree, where $h_j^k$ is the height of the tree $\mathcal{T}_k$ before the assignment of the delivery $I_j$. All the other update operations (lines 6, 7, and 10) are standard operations on a balanced binary search tree, and so need $\mathcal{O}(h_j^k)$ time to assign the delivery $I_j$ in $\J_k$. Consequently, for the interval set $\J_k$ Algorithm \ref{GreedyAlgo_mod} runs in $\sum_{I_j \in J_k} \mathcal{O}(h_j^k) \leq \mathcal{O}(n_k \log n_k)$ time, where $n_k$ is the number of deliveries in $\J_k$. Therefore, the total running time of the algorithm is $\sum_{k=1}^{\omega} \mathcal{O}(n_k \log n_k) = \mathcal{O}(n \log n)$.
In summary, overall running time for Algorithm \ref{Alg:approx_ddp_color} is $\mathcal{O}(n \log n + n_e)$.  
\end{proof}
\begin{theorem}
\label{Theorem-3}
Algorithm \ref{Alg:approx_ddp_color} is an approximation algorithm for DDP-NS, which uses at most $\big(\frac{OPT_{NS}}{(1-\epsilon_{max})} + \omega (1 - \frac{\epsilon_{min}}{1 - \epsilon_{max}})\big)$ drones, where $\omega$ is the maximum clique size of $G$.
\end{theorem}
\begin{proof}
The correctness and the polynomial running time proof of the algorithm follow from the Lemma \ref{lem:cor+compl-ddp-ns}.
Let $cost(\cI)$ be the total cost of all the deliveries in $\cI$ and $cost(\J_k)$ be the total cost of all the deliveries in $\J_k (1 \leq k \leq \omega)$. 
 Then, $\sum\limits_{k=1}^{\omega} cost(\J_k) = cost(\cI) \leq OPT_{NS} \cdot B$.\\
 If $m_k$ is the number of drones returned by the Algorithm \ref{GreedyAlgo_mod} for the interval set $\J_k$, then by the Corollary \ref{cor:weight}, $\sum_{k=1}^{\omega}((m_k - 1) (1 - \epsilon_{max})B + \epsilon_{min} B) \leq \sum\limits_{k=1}^\omega cost(\mathcal{J}_k) \leq OPT_{NS} \cdot B$. Hence
\begin{equation}
  \sum_{i=1}^{\omega} m_k \leq \frac{OPT_{NS}}{(1-\epsilon_{max})} + \omega (1 - \frac{\epsilon_{min}}{1 - \epsilon_{max}}). \label{eq-14}
  \end{equation}
%Hence the proof.  
\end{proof}
\begin{corollary}
\label{cor:genralized_factor}
    Algorithm \ref{Alg:approx_ddp_color} is a ($2 + \frac{\epsilon_{max}-\epsilon_{min}}{1-\epsilon_{max}}$)-approximation algorithm for DDP-NS.
    %, where $\psi = \frac{\epsilon_{max}-\epsilon_{min}}{1-\epsilon_{max}}$.
    %If ($\epsilon_{max} - \epsilon_{min}) \leq \frac{\theta}{2}$, then Algorithm \ref{Alg:approx_ddp_color} is a ($2 + \theta$) approximation algorithm for DDP-NS.
\end{corollary}
\begin{proof}
  The statement follows from the Eq. \ref{eq-14} and by using $\omega \geq OPT_{NS}$ from Eq. \ref{eq:rel-w-opt}.
\end{proof}
% \begin{corollary}
%     If ($\epsilon_{max} - \epsilon_{min}) \leq \frac{\theta}{2}$, then Algorithm \ref{Alg:approx_ddp_color} is a ($2 + \theta$) approximation algorithm for DDP-NS.
%     \label{cor:consolidate-factor-ddp-ns}
% \end{corollary}
% \begin{proof}
%     This holds from Corollary \ref{cor:genralized_factor} and by using $\epsilon_{max} \leq \frac{1}{2}$.
% \end{proof}

\begin{discussion}
    The above result suggests that if $\epsilon_{min}$ and $\epsilon_{max}$ are close to each other, then we can obtain an approximation algorithm with a factor closer to 2.
    %We have $\epsilon_{min} = \frac{1}{B}\min_{j \in \cN} cost(I_j)$ and $\epsilon_{max} = \min\{\frac{1}{2}, \frac{1}{B}\max_{j \in \cN}cost(I_j)\}$.
    $\epsilon_{min}$ and $\epsilon_{max}$ being close implies that the difference between the cost of any two deliveries is small. In other words, all the distances between the location of customers and the truck's route are nearly similar, or the customers are uniformly distributed around the truck's route. Moreover, the above approximation factor corresponds to the worst-case scenario where $\omega = OPT_{NC}$. In practice, this situation is unlikely to occur. The number of deliveries per day is usually much larger than the maximum number of deliveries active at any given time. Consequently, in most practical instances $OPT_{NC}$ is significantly larger than $\omega$. Thus, expressing our results in terms of $\omega$ is meaningful and may yield a smaller approximation factor in practical settings.
\end{discussion}

\begin{discussion}
    \label{dis:worst-case-discussion}
    %In the worst case, when $\epsilon_{min}$ is close to zero and $\epsilon_{max}$ is close to half, the algorithm may perform poorly. In other words, there may be a delivery very close to the truck and another delivery that consumes at least half of the drone's capacity.
    We can relate DDP-NC with the \emph{bin packing with conflicts} (BPPC) problem~\cite{epstein}. Epstein and Levin proposed a $\frac{7}{3}$-approximation algorithm for BPPC when the conflict graph is an interval graph. Therefore, their algorithm is also applicable to DDP-NC. Our proposed Algorithm~\ref{Alg:approx_ddp_color} performs better than the reult in \cite{epstein} only when $\psi < \frac{1}{3}$.
    However, Algorithm~\ref{Alg:approx_ddp_color} becomes useful in designing algorithms for DDP-SC (Sections~\ref{sec:ddp-sc} and~\ref{sec:mod-ddp-sc}). If the $\frac{7}{3}$-approximation algorithm \cite{epstein} is used directly as a subroutine for DDP-SC, it leads to a higher approximation ratio than that obtained by our proposed algorithm (Theorem~\ref{thm:final-ddp-sc-mod}). We have discussed this fact again in Discussion \ref{dis:motivation-for-modification}.
\end{discussion}
\begin{discussion}
\label{dis:approx-ddp-color-2}
    The analysis of the algorithm is independent of the order of the interval. Therefore, if we fix some of the ordering of the intervals, the analysis of the algorithm remains the same. We will exploit this property in designing an algorithm in Section \ref{sec:mod-ddp-sc}.
\end{discussion}

\section{Approximation Algorithm for DDP-NC}
\label{sec:ddp-nc}
Here we demonstrate \textsc{AlgoFor-DDP-NC}, an approximation algorithm for solving the \emph{Drone Delivery Packing Problem without Conflicting Intervals} (DDP-NC). An illustration of the algorithm, with an example, is shown in
%Section 2 of the supplemental file.
Section \ref{sec:ex-ddp-nc}, Appendix \ref{sec:appendix}.
We design the algorithm assuming battery service stations act as charging stations. The same approach extends to battery swapping.

We first partition the delivery interval set $\mathcal{I}$ into disjoint subsets 
$\{\mathcal{I}_1, \mathcal{I}_2, \ldots, \mathcal{I}_{r+1}\}$, 
where $r$ denotes the number of charging stations. 
For each subset $\mathcal{I}_{\ell}$ $(1 \le \ell \le r+1)$, 
we compute the required number of drones $m_{\ell}$ and construct the corresponding schedule using a standard bin-packing strategy. 
To minimize the total number of drones used, we recharge drones whenever possible and prioritize the reuse of previously used drones.
We now describe the algorithm step by step, providing key observations after each step to justify its correctness.%\newline

\textbf{Step-1:} \textit{~Divide the interval set $\cI$ into disjoint subset of intervals $\{\mathcal{I}_1, \mathcal{I}_2, \ldots, \mathcal{I}_{r+1}\}$, where
% $\cI_1 = \{I_j ~|~ t_j^L < t_1^A\}$, and for
% $2 \leq \ell \leq r+1,~ \cI_{\ell} = \{I_j ~|~  I_j\notin \bigcup\limits_{\ell'=1}^{\ell-1} \cI_{\ell'} \text{~and~} (t_j^L < t_{\ell}^A)\}$.
$\cI_1 = \{I_j ~|~ t_j^L < t_1^A\}$, and for 
$2 \leq \ell \leq r+1, ~ \cI_{\ell} = \{I_j ~|~ t_{\ell -1}^A \leq t_j^L < t_{\ell}^A\}$, and $\cI_{r+1} = \{I_j ~|~ t_j^L \geq t_{\ell}^A\}$.} %\newline

%We define $t_{r+1}^A > \max\limits_{j \in \cN} t_j^R$.
It is immediate to follow that each delivery interval $I_j$ belongs to one of the subsets $\cI_{\ell}$. However, some $\cI_{\ell}$ may contain no intervals. For $\ell \geq 2$, the interval of $\cI_{\ell}$ containing $t_{\ell-1}^D$ is denoted by $I_{\ell}^{first}$, and for $\ell \leq r$, the interval of $\cI_{\ell}$ containing $t_{\ell}^A$ is denoted by $I_{\ell}^{last}$. For some $\ell$, $I_{\ell}^{first}$ or $I_{\ell}^{last}$ or both may not exist.
Because the intervals are conflict-free and no interval is fully contained in a waiting interval, among all the intervals of $\cI_{\ell}$, only the interval $I_{\ell}^{first}$ (if exist) intersect with the waiting time interval $I_{\ell-1}^c$, and only the interval $I_{\ell}^{last}$ (if exist) intersect with the waiting time interval $I_{\ell}^c$. All other intervals of $\cI_{\ell}$ do not intersect with any of the waiting intervals.
In the next step, we further partition each set $\mathcal{I}_{\ell}$ 
$(1 \le \ell \le r+1)$ into blocks $\{S_{\ell}^i\}_i$. A \emph{block} is defined as a set of intervals whose total cost is at most $B$. We refer to this subproblem as \texttt{Partition}$(\mathcal{I}_{\ell})$. For each $\ell$ $(1 \le \ell \le r+1)$, we solve \texttt{Partition}$(\mathcal{I}_{\ell})$ by following FFD$(\mathcal{I}_{\ell})$, which is the classical \emph{First-Fit Decreasing} (FFD) bin-packing strategy. For completeness, we describe the strategy below using the terminology of blocks and intervals, rather than items and bins. 
 %\newline

\noindent \textbf{Step-2 (FFD$(\mathcal{I}_{\ell})$):} \textit{~We first sort the intervals of $\cI_{\ell}$ according to non-increasing order of their costs. Let $\cI_{\ell} =\{I_{\ell}^1, I_{\ell}^2, \cdots, I_{\ell}^p\}$ and without loss of generality let the intervals in $\cI_{\ell}$ are already sorted. We begin with $m_{\ell}=1$ and $S_{\ell}^1 = \{I_{\ell}^1\}$. To assign the remaining intervals, we select one by one in sorted order. To assign the interval $I_{\ell}^j$, we check whether there exists a feasible $i \le m_{\ell}$ such that $cost(S_{\ell}^{i}) + cost({I_\ell^j}) \le B$. If there exists multiple such feasible $i$, we choose the $i$ with the minimum index and add $I_{\ell}^j$ to $S_{\ell}^{i}$. If no such $i$ is available, we increase $m_{\ell}$ by one and set $S_{\ell}^{m_{\ell}} = \{I_{\ell}^j\}$.}%\newline

%We call each of the subsets $S_\ell^i$ a \emph{block}. Formally, a \emph{block} is the set of intervals with cost at most $B$.
The block assigned for the interval $I_{\ell}^{first}$ is named $S_{\ell}^{first}$. Whereas, the block assigned for the interval $I_{\ell}^{last}$ is named $S_{\ell}^{last}$. Since the intervals in a block satisfy both the compatibility and budget constraint, they can all be assigned to a drone having full battery capacity.
To reduce the number of used drones, a drone assigned for one of the blocks of $\mathcal{I}_{\ell}$ is also assigned for the subsequent blocks of $\mathcal{I}_{\ell'}$ with $\ell' > \ell$. %\newline%We denote $m_{max}$ to represent the $\max\limits_{1\le \ell \le r+1} m_{\ell}$.

\noindent \textbf{Step-3:} \textit{~Find $m_{max} = \max\limits_{1\le \ell \le r+1} m_{\ell}$. Open $(m_{max} + 2)$ drones. 
%We denote the drones by $\{D_1, D_2, \cdots, D_{m_{max}+2}\}$.
}%\newline

We now assign the blocks to these $m_{{max}+2}$ drones so that the assignment of each drone remains feasible at each time. Note that a drone assigned to the blocks of $\mathcal{I}_1$ may not be scheduled for $S_2^{first}$. This occurs when $cost(S_2^{first})$ and each $cost(S_1^i)$ are sufficiently large such that $cost(S_2^{first}) > Ch_2(rem_1^i, I)$, where $I \subseteq [t_1^A, t]$ and $t$ denotes the launch time of the interval $I_2^{first}$. In this case, no drone assigned to $\mathcal{I}_1$ has sufficient time to recharge so that it can be scheduled to $S_2^{first}$. Similarly, a drone assigned to blocks of $\mathcal{I}_{\ell}$ $(\ell \ge 2)$ may not be feasible for $S_{\ell+1}^{first}$. The same restriction applies to the drone $D$ assigned to $S_{\ell-1}^{last}$, since the interval $I_1^{last}$ may end close to $t_1^D$, while the interval $I_3^{first}$ may begin close to $t_2^A$. In view of such a worst-case scenario, we design our next step.
A drone assign to the block $S_{\ell}^{first}$ ($\ell \ge 2$) is referred to as $D_{\ell}^{first}$. Similarly, drone assign to the block $S_{\ell}^{last}$ ($\ell \le r$) is referred to as $D_{\ell}^{last}$. For completeness, we define $D_0^{last} = D_{-1}^{last} = \phi$, and $\cI_0 = \emptyset$.
%Additionally, $S_1^{first}$ is the first block returned 
%\newline

\noindent \textbf{Step-4:} \textit{~The blocks of $\cI_1$ are assigned uniquely to any of the $m_1$ drones.
For each $2 \le \ell \le r+1$, assign the block $S_{\ell}^{first}$ to one of the drones except $D_{\ell-2}^{last}$ and the drones assigned to the blocks of $\cI_{\ell-1}$. For the remaining blocks of $\cI_{\ell}$, we assign them uniquely to the drones except $D_{\ell-1}^{last}$ and $D_{\ell}^{first}$. If $\ell \le r$, after the assignment, charge all the used drones (having remaining batteries less than $B$) except $D_{\ell}^{last}$.}

We now analyze our algorithm in the following.
\begin{lemma}
\label{lem:cor-ddp-nc}
    \textsc{AlgoFor-DDP-NC} assigns each of the deliveries in $\cI$ to one of the $(m_{{max}}+2)$ drones, and all the assignments of those drones are feasible.
\end{lemma}
\begin{proof}
   From Step-1, each delivery belongs to a unique interval set $\cI_{\ell}$ for some $1 \le \ell \le r+1$. Further, by the correctness of FFD (Step-2), each delivery in $\cI_{\ell}$ belongs to a unique block and is therefore assigned to a unique drone.

    We prove the feasibility of the drone assignments by induction on $\ell$, where $\ell$ denotes the index of the $\ell$-th interval set $\cI_{\ell}$.
    For $\ell = 1$, the $m_1$ $(\leq m_{\max})$ blocks of $\cI_1$ are uniquely assigned to $m_1$ drones, since $D_1^{last} = \phi$ and $\cI_0 = \emptyset$. As all drones initially have empty assignments, the feasibility condition is trivially satisfied after this assignment.

Now assume, as the induction hypothesis, that the assignments of all drones are feasible up to the processing of the interval set $\cI_{\ell-1}$. Consider the processing of $\cI_{\ell}$. The block $S_{\ell}^{first}$ is assigned to a drone other than $D_{\ell-2}^{last}$ and the drones assigned to the blocks of $\cI_{\ell-1}$. Note that the intervals of $\cI_{\ell-1}$ and the interval $I_{\ell-2}^{last}$ are the only intervals that intersect the waiting interval $I_{\ell-2}^c$. Therefore, the assignments of all drones except $D_{\ell-2}^{last}$ and the drones assigned to the blocks of $\cI_{\ell-1}$ do not contain any interval intersecting $I_{\ell-2}^c$. Consequently, these drones can fully recharge at the $(\ell-2)$-th station.

Since the number of blocks of $\cI_{\ell-1}$ returned by Step~2 of the algorithm is at most $m_{\max}$, the total number of drones assigned to $\cI_{\ell-1}$, including $D_{\ell-2}^{last}$, is at most $(m_{\max}+1)$. Hence, there exists a drone with full capacity before processing the intervals in $\cI_{\ell}$, and assigning this drone to the block $S_{\ell}^{first}$ does not violate feasibility.

Furthermore, $I_{\ell-1}^{last}$ and $I_{\ell}^{first}$ are the only intervals that intersect the waiting interval $I_{\ell-1}^c$. Therefore, except for $D_{\ell-1}^{last}$ and $D_{\ell}^{first}$, all the other $(m_{\max}-1)$ drones can fully recharge at the $(\ell-1)$-th station. Assigning these drones uniquely to the remaining at most $(m_{\max}-1)$ blocks of $\cI_{\ell}$ preserves the feasibility condition.
\end{proof}

\begin{lemma}
    \label{lem:rel-opt-l-opt}
    For any $\ell : 1 \le \ell \le r+1$, $OPT_{\ell} \leq OPT_{NC}$, where $OPT_{\ell}$ is the number of blocks in an optimum partition of $\cI_{\ell}$, and $OPT_{NC}$ is the optimum solution size for DDP-NC.
\end{lemma}

\begin{proof}
    %On the contrary, suppose that for some $\ell$, we have $OPT_{\ell} > OPT_{NC}$. 
    From the partitioning process, the set $\mathcal{I}_{\ell}$ contains intervals $I_j$ with $t_{\ell-1}^A \le t_j^L < t_{\ell}^A$.
    Since no delivery interval is fully contained within any waiting interval, every interval in $\mathcal{I}_{\ell}$ has rendezvous time greater than $t_{\ell-1}^D$. 
    Consequently, a drone cannot have an assignment that contains the intervals from $\mathcal{I}_{\ell}$, then recharge at a station, and then again assigned further intervals from $\mathcal{I}_{\ell}$.
    Therefore, the total cost of intervals from $\mathcal{I}_{\ell}$ assigned to any single drone is at most $B$. 
    This property also holds for every drone in an optimal solution.

    Let $\{d_1^*, d_2^*, \ldots, d_{OPT_{NC}}^*\}$ denote the drones in an optimal solution of DDP-NC, and let $\mathcal{A}_i^*$ denote the assignment of drone $d_i^*$, where $1 \le i \le OPT_{NC}$.
    We define $S_{\ell}^{i^*} = \mathcal{A}_i^* \cap \mathcal{I}_{\ell}$. Then each non-empty set $S_{\ell}^{i^*}$ forms a block, and the collection $\{ S_{\ell}^{i^*} \}_{i=1}^{OPT_{NC}}$ constitutes a partition of $\mathcal{I}_{\ell}$.
    $OPT_{\ell}$ being the size of the optimum partition of $\cI_{\ell}$, $OPT_{\ell} \le OPT_{NC}$.
\end{proof}

\begin{theorem}
\label{thm:alg-ddp-nc}
   \textsc{AlgoFor-DDP-NC} is an approximation algorithm for DDP-NC that uses at most ($\alpha \cdot OPT_{NC} +2$) drones, where $\alpha$ is the best known approximation factor for FFD.% and $OPT_{NC}$ is the optimum solution for DDP-NC.
\end{theorem}
\begin{proof}
    The correctness of the algorithm holds from Lemma \ref{lem:cor-ddp-nc}. The interval partitioning can be done in linear time. Then applying FFD to each of the subset takes total $O(n \log n + r)$ time, as the intervals need to be sorted to each of those partition which takes $O(n_{\ell} \log n_{\ell} + 1)$ time, where $n_{\ell}$ is the size of the subset $\cI_{\ell}$. Thereafter, assigning blocks to the drones take at most $O(n)$ time. Hence, the overall time complexity of the algorithm is $O(n \log n +r)$.

    Let $OPT_{\ell}$ be the number of blocks in the partition of $\cI_{\ell}$, where $1 \le \ell \le r+1$.
    Then, $m_{\ell} \leq \alpha \cdot OPT_{\ell}$ and $OPT_{\ell} \leq OPT_{NC}$, $\forall ~1 \le \ell \le r+1$, from Lemma \ref{lem:rel-opt-l-opt}. So, if $\ell^* = \arg \max\limits_{1\leq \ell \leq r+1} m_{\ell}$, then $m_{max} = m_{\ell^*} \leq \alpha 
    OPT_{\ell^*}$.
    % \begin{equation}
    % \label{eq:rel-m-opt-nc}
    %     m_{max} = m_{\ell^*} \leq \frac{3}{2}OPT_{NC}
    % \end{equation}
    Hence, the number of drones used by the algorithm is,
    $m_{max} + 2 = m_{\ell^*} + 2 \leq \alpha OPT_{\ell^*} + 2 \leq \alpha OPT_{NC} + 2$
\end{proof}

\begin{discussion}
\label{disc:ddp-nc-1}
The best-known asymptotic approximation ratio for FFD is
$\frac{11}{9}OPT_{NC} + \frac{6}{9}$ \cite{tightFactorFFD}. Therefore, based on the above theorem, our algorithm uses at most $\frac{11}{9}OPT_{NC} + \frac{24}{9}$ drones. This bound is tight, as there exists an instance of the bin-packing problem that matches the FFD bound \cite{tightFactorFFD}.
More generally, any bin-packing algorithm can be used in place of FFD in Step-2, and the corresponding approximation factor can be incorporated into the analysis.
It is also known that the absolute approximation ratio of FFD is $\frac{3}{2}$ \cite{simchi1994new}, and this bound is tight. Moreover, unless $P=NP$, no polynomial-time algorithm for bin packing can achieve an absolute approximation ratio better than $\frac{3}{2}$, which follows from the hardness of the \textit{partition problem} \cite{garey2002computers}. 
Thus, using FFD in Step-2 yields a solution of size at most $\frac{3}{2}OPT_{NC} + 2$. The additive constant can be reduced by $\frac{1}{2}$ through a modification of the algorithm, as described below. 
%The detailed modified algorithm can be found in \cite{}.
In the modified algorithm, we use the fact that if for any $\ell: 1 \le \ell \le r$, $m_{\ell} = \frac{3}{2}OPT_{\ell}$ holds then for any block $S_{\ell}^i$ of $\cI_{\ell}$ in our solution, $cost(S_{\ell}^i) \leq cost(T_{\ell}^j)$ must hold, where $T_{\ell}^j$ is one of the blocks of $\cI_{\ell}$ in the optimum solution.
Using this result, we modify our algorithm so that we can use a drone assigned to one of the blocks of $\cI_{\ell}$ to the block $S_{\ell+1}^{first}$ as well.
Consequently, the number of used drones becomes at most $\frac{3}{2}OPT_{NC} + \frac{3}{2}$.
It is worth noting that the bound $\frac{3}{2}OPT_{NC} + \frac{3}{2}$ is lesser than the asymptotic bound $\frac{11}{9}OPT_{NC} + \frac{24}{9}$ for small values of $OPT_{NC}$.
So, the modification does not yield a better solution for bigger optimum values.
Nevertheless, we believe that our proposed modification and the corresponding analysis may be useful for other results in the literature that use FFD as their subroutine.
%Nevertheless, we present the following modified algorithm to provide intuition that will be useful in designing and analyzing the algorithm for DDP-SC in the next section.
\end{discussion} 

\subsection{Modification to \textsc{AlgoFor-DDP-NC}}
In the previous algorithm, we have seen that, in the worst case, it may be possible that we can not assign any of the drones to the block $S_{\ell+1}^{first}$ that are already assigned to the blocks of $\cI_{\ell}$ and to $S_{\ell -1}^{last}$. Therefore, if $m_{\ell}= m_{max}$, then the algorithm needs $(m_{max}+2)$ drones for the blocks $\{S_{\ell-1}^{last}, S_{\ell}^1, S_{\ell}^2, \cdots, S_{\ell}^{m_{max}}, S_{\ell+1}^{first}\}$. Moreover if $m_{\ell} = \frac{3}{2}OPT_{\ell}$, we have the solution size at most $\frac{3}{2}OPT_{NC} + 2$. However, in the lemma below, we prove that, if $m_{\ell} = \frac{3}{2}OPT_{\ell}$, then for any block $S_{\ell}^i$ of $\cI_{\ell}$ in our solution, $cost(S_{\ell}^i) \leq cost(T_{\ell}^j)$ must hold, where $T_{\ell}^j$ is one of the blocks of $\cI_{\ell}$ in the optimum solution. By using this result, we modify our algorithm so that we can use a drone assigned to one of the blocks of $\cI_{\ell}$ to $S_{\ell+1}^{first}$ as well.

\begin{lemma}
    \label{lem:property-ffd}
    If there exists a block $S_{\ell}^i$ such that $cost(S_{\ell}^i) > cost(T_{\ell}^j)$, then $m_{\ell} \leq \frac{3}{2}OPT_{\ell} -\frac{1}{2}$, where $S_{\ell}^i$ is one of the blocks produced by the algorithm FFD$(\cI_{\ell})$ and $T_{\ell}^j$ is one of the blocks in the optimum solution of the problem \texttt{Partition}$(\cI_{\ell})$.
\end{lemma}

\begin{proof}
   To prove the lemma, it suffices to show that $m_{\ell} \neq \frac{3}{2}OPT_{\ell}$, since it is already known that $m_{\ell} \le \frac{3}{2}OPT_{\ell}$ \cite{simchi1994new}.
    If $OPT_{\ell}$ is odd, then $m_{\ell}$ cannot equal $\frac{3}{2}OPT_{\ell}$. 
    Hence, it remains to consider the case where $OPT_{\ell}$ is even. On the contrary, we assume that $m_{\ell} = \frac{3}{2}OPT_{\ell}$.
    Let $OPT_{\ell} = 2t$ for some integer $t \ge 1$. Then $m_{\ell} = 3t$. Let $\{S_{\ell}^1, \ldots, S_{\ell}^{3t}\}$ denote the blocks produced by FFD$(\mathcal{I}_{\ell})$, and let $\{T_{\ell}^{1}, \ldots, T_{\ell}^{2t}\}$ denote the blocks in an optimal solution of \texttt{Partition}$(\mathcal{I}_{\ell})$. 
    For simplicity, let $c_{i'} = cost(S_{\ell}^{i'})$ and $d_{j'} = cost(T_{\ell}^{j'})$, where $1 \le i' \le 3t$ and $1 \le j' \le 2t$. Since each block has total cost at most $B$, we have $c_{i'}, d_{j'} \le B$. Moreover,
    \[
    \sum_{i'=1}^{3t} c_{i'} = \sum_{j'=1}^{2t} d_{j'}.
    \]
   From the lemma condition, if there exists $i$ and $j$ such that $c_i > d_j$. Then
    \begin{equation}
    \label{eq:upper-bound-c-i}
    \sum_{i'=1}^{3t} c_{i'} - c_i < \sum_{j'=1}^{2t} d_{j'} - d_j \le (2t-1)B.
    \end{equation}
By the defining property of FFD, for any block $S_{\ell}^{i'}$ and any $I \in S_{\ell}^{i''}$ with $i'' > i'$, we have
    \begin{equation}
    \label{eq:ffd-prop}
    c_{i'} + cost(I) > B.
    \end{equation}
Let $I_{\ell}^{i'}$ denote the first interval assigned to block $S_{\ell}^{i'}$. (Case-1) $i > 2t$, and (Case-2) $i \le 2t$. For each of the case, we establish some contradiction. We further divide Case-1.

\textbf{Case-1.1:} ($i > 2t$ and $cost(I_{\ell}^{2t}) \le \frac{B}{2}$): Because FFD orders intervals in non-increasing cost, every interval in blocks $S_{\ell}^{i'}$ with $i' > 2t$ also has cost at most $\frac{B}{2}$. Hence, each such block (except possibly $S_{\ell}^{3t}$) must contain at least two intervals.
Let $\mathcal{S} = \{S_{\ell}^{2t}, S_{\ell}^{2t+1}, \cdots, S_{\ell}^{i-1}, S_{\ell}^{i+1}, \cdots, S_{\ell}^{3t}\}$
Then $\mathcal{S}$ contains at least $(2t+1-2) = (2t-1)$ intervals. Now consider an arbitrary set of $(2t-1)$ intervals of $\mathcal{S}$, say $\{I_1', I_2', \cdots, I_{2t-1}'\}$. Then, by using Eq. \ref{eq:ffd-prop}, 
\begin{align*}
    & \sum_{i'=1}^{2t-1} \big(c_{i'} + cost(I_{i'}')\big) > (2t-1)B \\
    &  \implies  \sum_{i'=1}^{3t} c_{i'} - c_i > (2t-1)B,
    \text{which is a contradiction with  Eq. \ref{eq:upper-bound-c-i}}.
\end{align*}
% \[\sum_{i'=1}^{2t-1} \big(c_{i'} + cost(I_{i'}')\big) > (2t-1)B \implies  \sum_{i'=1}^{3t} c_{i'} - c_i > (2t-1)B, \text{~contradicting  Eq. \ref{eq:upper-bound-c-i}.}\]
\textbf{Case-1.2}: ($i > 2t$ and $cost(I_{\ell}^{2t}) > \frac{B}{2}$): In this case, if $cost(I_{\ell}^{2t+1}) > \frac{B}{2}$, then the first $(2t+1)$ intervals each exceed $\frac{B}{2}$.
Thus any feasible partition requires at least $(2t+1)$ blocks, a contradiction. Whereas, if $cost(I_{\ell}^{2t+1}) \leq \frac{B}{2}$, then the set $\{S_{\ell}^{2t+1}, S_{\ell}^{2t+2}, \cdots, S_{\ell}^{i-1}, S_{\ell}^{i+1}, \cdots, S_{\ell}^{3t}\}$ contains at least $(2t-1-2) = (2t-3)$ intervals. In addition if the block $S_{\ell}^{2t}$ contains two intervals, then the set $\mathcal{S}$ have at least $(2t-1)$ intervals. Hence, an argument identical to Case-1.1 yields a contradiction.
On the other hand, if $S_{\ell}^{2t}$ contains only one interval, then
\begin{align*}
 &~~~~~~~ cost(I_{\ell}^{2t}) + cost(I_{\ell}^{2t+1}) > B  \\
 & \implies cost(I_{\ell}^{i'}) + cost(I_{\ell}^{2t+1}) > B, \forall~ 1\leq i' \leq 2t.
\end{align*}
% \[cost(I_{\ell}^{2t}) + cost(I_{\ell}^{2t+1}) > B \implies cost(I_{\ell}^{i'}) + cost(I_{\ell}^{2t+1}) > B, \forall~ 1\leq i' \leq 2t.\]
Thus the first $(2t+1)$ intervals must occupy distinct blocks in any feasible solution,
requiring at least $(2t+1)$ blocks, which is a contradiction.

 \textbf{Case-2:} ($i \leq 2t$:) In this case, if $cost(I_{\ell}^{2t+1}) \le \frac{B}{2}$, the set $\{S_{\ell}^{2t+1}, S_{\ell}^{2t+2}, \cdots, S_{\ell}^{3t}\}$ contains at least $(2t-1)$ intervals, say $\{I_1'', I_2'', \cdots, I_{2t-1}''\}$. Then, $\forall 1\le i' (\neq i) \leq 2t, ~ c_{i'} + cost(I''_{i'}) > B$, implies \[\sum_{\substack{i'=1\\i'\neq i}}^{2t} c_{i'} + \sum_{\substack{i'=2t+1}}^{3t} c_{i'} > (2t-1)B,  \text{~contradicts  Eq. \ref{eq:upper-bound-c-i}.}\]
If instead $cost(I_{\ell}^{2t+1}) > \frac{B}{2}$,
then at least $(2t+1)$ blocks are required in any solution, which is again a contradiction. \\
Hence, in all cases we must have, $m_{\ell} \leq \frac{3}{2}OPT_{\ell} -\frac{1}{2}$.
\end{proof}

We now describe the step of the modified algorithm, which we name as \textsc{Mod-Algo-DDP-NC}. Here, we assume that $OPT_{NC} \geq 2$. This is because, if $OPT_{NC} = 1$, then the total cost of all the intervals in any  $\cI_{\ell} ~(1\leq \ell \leq r+1)$ is at most $B$, and so $m_{\ell} =1$ as well. Thus, $m_{max} + 2 = 3 \leq \frac{3}{2}OPT + \frac{3}{2}$ holds. Also $m_{max}$ must be at least $3$, as $OPT_{NC} \geq 2$ and $m_{max}=2$, implies $m_{max} + 2 \leq \frac{3}{2}OPT + \frac{3}{2}$.
We further assume that for any $\ell$ with $m_{\ell} = m_{max}$, $m_{\ell} = \frac{3}{2}OPT_{\ell}$ must hold, otherwise we again have $m_{max}+2 \le \frac{3}{2}OPT_{NC} + \frac{3}{2}$.
The $i$-th step of the modified algorithm is described in Step-$i^+$.%\newline

\textbf{Step-$1^+$:} \textit{Similar to Step-1, we first divide the set $\cI$ into subset of intervals $\{\cI_1, \cI_2, \cdots, \cI_{r+1}\}$. Then, for each $1 \le \ell \le r+1$, we partition $\mathcal{I}_{\ell}$ into collection of blocks $\{S_{\ell}^i\}$ by applying FFD$(\mathcal{I}_{\ell})$. Then, find $m_{max}  = \max\limits_{1 \le \ell \le r+1} m_{\ell}$, where $m_{\ell}$ is the number of blocks returned by FFD$(\mathcal{I}_{\ell})$.}%\newline

In the next step, we modify the cost of some intervals. This modification may increase the size of the partition. 
We denote by $m_{\ell}^+$ the size of the modified partition of $\cI_{\ell}$. We initialize $m_1^+ = m_1$. Then, for each $2 \le \ell \le r+1$, we apply the modification procedure sequentially, starting with $\ell = 2$. Let $t'_{\ell}$ be the time very close to the launch time of $I_{\ell}^{first}$, but less than it.%\newline

 \textbf{Step-$2^+$:} \textit{If $m_{\ell -1} \leq (m_{max}-1)$, we keep the partition for $\cI_{\ell}$ unchanged and set $m_{\ell}^+ = m_{\ell}$. Otherwise ($m_{\ell -1} \geq m_{max}$), select a block $S_{\ell -1}^i$ from the previous partition that does not contain the intervals $I_{\ell-1}^{first}$ and $I_{\ell-1}^{last}$. 
Modify the cost of $I_{\ell}^{first}$ by $\big(cost(I_{\ell}^{first}) + B -  Ch_{\ell-1}(B-cost(S_{\ell-1}^i), I')\big)$. We denote the modified interval $I_{\ell}^{first}$ by $\widehat{I}_{\ell}^{first}$.
We replace the earlier partition FFD$(\cI_{\ell})$ by FFD$(\cI_{\ell} \setminus \{I_{\ell}^{first}\} \cup \widehat{I}_{\ell}^{first})$; set $m_{\ell}^+$ as the size of this modified partition; and the block $S_{\ell-1}^i$ is referred to as $S_{\ell-1}^{spl}$.}%\newline

%We refer $S_{\ell}^{first}$ to represent the block of $\cI_{\ell}$ that contain either the dummy interval $I$ or $I_{\ell}^{first}$.
%We aim to use the same drone for the blocks $S_{\ell}^{first}$ and $S_{\ell-1}^{spl}$.\newline 

 \textbf{Step-$3^+$:} \textit{Find $m_{max}^+ = \max\limits_{1 \le \ell \le r+1} m_{\ell}^+$. used $(m_{max}^+ +1)$ drones.}

 \textbf{Step-$4^+$:} \textit{For each $\ell$, if $m_{\ell -1}^+ \leq (m_{max}^+-1)$, then consider the blocks of $\cI_{\ell}$ from the FFD$(\cI_{\ell})$. In case of $m_{\ell -1}^+ = (m_{max}^+-1)$, then consider the blocks of $\cI_{\ell}$ from FFD$(\cI_{\ell} \setminus \{I_{\ell}^{first}\} \cup \widehat{I}_{\ell}^{first})$. If $m_{\ell -1}^+ \leq (m_{max}^+-1)$, assign $S_{\ell}^{first}$ to one of the drones other than $D_{\ell-2}^{last}$ and the drones assigned to the blocks of $\cI_{\ell-1}$.
Whereas if $m_{\ell-1}^+ = m_{max}^+$, we assign the block $S_{\ell}^{first}$ to the drone that was assigned for $S_{\ell-1}^{spl}$. In either case, the remaining blocks of $\cI_{\ell}$ are assigned to those drones except the drones $D_{\ell}^{first}$ and $D_{\ell-1}^{last}$. After assignment, if $\ell \leq r$, we charge all the used drones (having remaining capacity less than $B$) for the entire waiting interval $I_{\ell}^c$, except the drone $D_{\ell}^{last}$ and the drones assigned to $S_{\ell}^{spl}$ (if exists). If $S_{\ell}^{spl}$ exists, the drone assigned to it is charged during the interval $[t_{\ell-1}^A, t']$, where $t' = \sup \limits_{t < t''} t$ and $t''$ the launch time of the interval $I_{\ell+1}^{first}$.}%\newline

We now analyze this modified algorithm.
\begin{lemma}
    \label{lem:rel-m-m+}
    $m_{max} \leq m_{max}^+ \leq m_{max}+1$.
    %, where $m_{max}$ and $m_{max}^+$ are the maximum size of the partition in Step-2 and Step-2+, respectively.
\end{lemma}

\begin{proof}
    In the modification step for $\mathcal{I}_{\ell}$, we alter the input to the partitioning procedure by increasing the cost of the interval $I_{\ell}^{first}$ only. Therefore, the size of the partition either remains the same or increases by one. Consequently, $m_{max}^+$ is either $m_{max}$ or $m_{max}+1$.
\end{proof}

\begin{lemma}
\label{lem:no-change-opt-after-mod}
    If an optimum solution to \texttt{Partition}$(\cI_{\ell})$ contains a block $S$ with cost at most $B' \leq B$, then the size of the optimum solution of \texttt{Partition}$(\cI_{\ell} \setminus \{I\} \cup \{\widehat{I}\})$ equals $OPT_{\ell}$, where $I$ is an interval in the block $S$ and $\widehat{I}$ is a copy of the  interval with $cost(\widehat{I}) = cost(I) + B - B'$.
\end{lemma}

\begin{proof}
We construct a solution for \texttt{Partition}$( (\mathcal{I}_{\ell} \setminus \{I\}) \cup \{\widehat{I}\} )$ from an optimal solution of \texttt{Partition}$(\mathcal{I}_{\ell})$ by modifying the block $S$. 
Specifically, we replace the interval $I$ in $S$ with $\widehat{I}$. The cost of the modified block is at most
$B' - cost(I) + cost(\widehat{I}) = B' -  cost(I) + cost(I)+ B - B' = B.$

Hence, the modified block remains feasible.
Therefore, the optimal value of \texttt{Partition}$(\mathcal{I}_{\ell} \setminus \{I\}) \cup \{\widehat{I}\})$ is at most $OPT_{\ell}$.
Conversely, by reversing the same argument, we can construct a feasible solution for \texttt{Partition}$(\mathcal{I}_{\ell})$ from an optimal solution of \texttt{Partition}$(\mathcal{I}_{\ell} \setminus \{I\}) \cup \{\widehat{I}\})$. Thus, both instances have the same optimal value, which completes the proof.
\end{proof}

To analyze the algorithm we now consider some assumptions.
\begin{assumption}
    For any $1 \le \ell \le r+1$ with $m_{\ell} = m_{max}$, we must have $m_{\ell} = \frac{3}{2}OPT_{\ell}$ and $OPT_{\ell} = OPT_{NC}$. Otherwise, the previous algorithm produces at most $\frac{3}{2}OPT_{NC} + \frac{3}{2}$ drones, which is sufficient to achieve our desired objective (Theorem~\ref{thm:main-ddp-nc}). Additionally, we assume that for any $\ell$ with $m_{\ell}^+ = m_{max}^+$, we must have $m_{\ell}^+ = \frac{3}{2}OPT_{NC}$. Otherwise, we can open $(m_{max}^++2)$ drones instead of $(m_{max}^++1)$ drones, which keeps the feasibility condition satisfied and we will achieve our desired result again.
    \label{asum:mod-ddp-nc}
\end{assumption}

\begin{lemma}
    \label{lem:feasibility-mod-ddp-nc}
    Under Assumption \ref{asum:mod-ddp-nc}, \textsc{Mod-Alg-DDP-NC} assigns each of the deliveries in $\cI$ to one of the $(m_{{max}}^+ +1)$ drones, and all the assignments of those drones are feasible. Moreover, for any $1\le \ell \le r+1$, $OPT_{\ell}^+ \leq OPT_{NC}$, where $OPT_{\ell}^+$ is the size of the optimum partition of $\cI_{\ell}$ after the modification.
\end{lemma}

\begin{proof}
    By a similar argument as of Lemma \ref{lem:cor-ddp-nc}, it is evident that the modified algorithm also assign each of the delivery to one of the drones. The rest of the lemma claim's by induction on $\ell$, where $\ell$ denotes the index of the $\ell$-th interval set $\cI_{\ell}$.
    For $\ell = 1$, there is no modification on $\cI_1$, and so all the assignment is feasible after the processing of $\cI_1$, by Lemma \ref{lem:cor-ddp-nc}. Moreover, $OPT_1^+ = OPT_1 \leq OPT_{NC}$, by Lemma \ref{lem:rel-opt-l-opt}. 

    We assume that, for any $1\le \ell' \le \ell-1$, $OPT_{\ell'} \le OPT_{NC}$, and the assignments of the drones are feasible up to the processing of the intervals of $\cI_{\ell-1}$. Now consider the interval set $\cI_{\ell}$. We prove the feasibility condition for the drone that was assigned to the block $S_{\ell}^{first}$. For all other blocks, we always have sufficient number of drones with full capacity after the time $t_{\ell-1}^D$, and hence assigning those blocks uniquely to the drones (as described in Step-$4^+$) keep the feasibility condition satisfied.
    If $m_{\ell-1}^+ \le m_{max}^+$, then the algorithm consider the partition of $\cI_{\ell}$ without modification. So, $OPT_{\ell}^+ = OPT_{\ell} \le OPT_{NC}$. In addition, when $m_{\ell-1}^+ \le m_{max}^+$, there is always a drone $D$ that is not assigned either to $\cI_{\ell-1}$ or to $S_{\ell-2}^{last}$. The drone must be of full capacity before processing of $\cI_{\ell}$, and so assigning $S_{\ell}^{first}$ to $D$ keep the feasibility condition satisfied.
    
    Now consider the case, $m_{\ell-1}^+ = m_{max}^+$. Then, from the induction hypothesis, $OPT_{\ell-1}^+ \le OPT_{NC}$. From Assumption \ref{asum:mod-ddp-nc}, we have $m_{\ell-1}^+ = \frac{3}{2}OPT_{NC}$. From the bound of FFD, we have $m_{\ell-1}^+ \leq \frac{3}{2}OPT_{\ell-1}^+$. All these three relations implies that,
    \[m_{\ell-1}^+ = \frac{3}{2}OPT_{\ell-1}^+ \text{~and~} OPT_{\ell-1}^+ = OPT_{NC}\]

    Let $S_{\ell-1}^{spl}$ be the block of $\mathcal{I}_{\ell-1}$ that does not contain $I_{\ell-1}^{first}$ or $I_{\ell-1}^{last}$. Let $D^*$ be the drone in the optimal solution assigned to $I_{\ell}^{first}$, and let $S^*$ denote the intervals assigned to $D^*$ in $\mathcal{I}_{\ell-1}$. Since $m_{\ell-1}^+ = \frac{3}{2}OPT_{\ell-1}^+$ and $OPT_{\ell-1}^+ = OPT_{NC}$, Lemma~\ref{lem:property-ffd} implies $cost(S_{\ell-1}^{spl}) < cost(S^*)$.

    Let $D$ be the drone assigned to $S_{\ell-1}^{spl}$ in our algorithm. Because $S_{\ell-1}^{spl}$ does not contain $I_{\ell-1}^{first}$, drone $D$ has full capacity before processing this block. 
    Hence, after the processing of $\cI_{\ell-1}$, drone $D$ has remaining capacity at least the capacity of $D^*$. Let $I'$ and $I''$ are the charging intervals of the drones $D$ and $D^*$ at the station $(\ell-1)$, respectively. From our algorithm, $I'=[t_{\ell-1}^A, t']$. The end-time for $I''$ is also at most $t'$. Then, before the processing of $\cI_{\ell}$, the remaining capacity of $D$ (say $rem$) is at least the remaining capacity of the drone $D^*$ (say $rem^*$). Therefore, assigning $D$ to $I_{\ell}^{first}$ (and so the block $S_{\ell}^{first}$) does not violate the feasibility condition.
    
    In our case, $rem = Ch_{\ell-1}(B-cost(S_{\ell-1}^i), I').$
    Since, the total cost of intervals in $\mathcal{I}_{\ell}$ assigned to $D^*$ is at most $rem$, by Lemma~\ref{lem:no-change-opt-after-mod}, the optimal value of \texttt{Partition}$(\mathcal{I}_{\ell})$ equals that of \texttt{Partition}$((\mathcal{I}_{\ell} \setminus \{I_{\ell}^{first}\}) \cup \{\widehat{I}_{\ell}^{first}\})$, where $cost(\widehat{I}_{\ell}^{first}) = cost(I_{\ell}^{first}) + B - rem$. In other words, $OPT_{\ell}^+ = OPT_{\ell}$. Hence, by Lemma \ref{lem:rel-opt-l-opt}, $OPT_{\ell}^+ \le OPT_{NC}$ holds. 
\end{proof}

\begin{lemma}
    Under Assumption \ref{asum:mod-ddp-nc}, \textsc{Mod-Alg-DDP-SC} uses at most $(\frac{3}{2}OPT_{NC} + 1)$ drones.
    \label{lem:bound-mod-ddp-nc}
\end{lemma}
\begin{proof}
    Let $\ell^* = \arg \limits_{1 \le \ell \le r+1} m_{\ell}^+$. Then, the number of used drones by the modified algorithm is 
    \begin{align*}
       m_{max}+1 & = m_{\ell^*}+1
       \le \frac{3}{2}OPT_{\ell^*} + 1,\\ %\text{~~from the bound on FFD}\\
       & \le \frac{3}{2}OPT_{NC} + 1, \text{~~from Lemma \ref{lem:feasibility-mod-ddp-nc}}
   \end{align*}
\end{proof}
By combining Theorem~\ref{thm:alg-ddp-nc} and Lemma~\ref{lem:bound-mod-ddp-nc}, together with a relaxation of Assumption~\ref{asum:mod-ddp-nc}, we obtain the following result.
\begin{theorem}
        \label{thm:main-ddp-nc}
        There exists an approximation algorithm for DDP-NC that uses at most\\ $\min\{\frac{11}{9}OPT_{NC} + 2, \frac{3}{2}OPT_{NC} + \frac{3}{2}\}$ drones.
    \end{theorem}

    \begin{discussion}
        A similar analysis applies when charging stations are replaced by swapping stations. A drone $D$ charges only partially at station $\ell$ when it is assigned to the block $S_{\ell+1}^{first}$. In case of the station being the swapping station, we modify the algorithm so that $D$ does not swap its battery at station $\ell$. In this case, a drone $D^*$ in the optimal solution assigned to $I_{\ell+1}^{first}$ also cannot swap its battery at station $\ell$. Consequently, $OPT_{\ell+1}^+ \le OPT_{NC}$, and Theorem~\ref{thm:main-ddp-nc} remains true for swapping stations.
    \end{discussion}

\section{Approximation Algorithm for DDP-SC}
\label{sec:ddp-sc}
Here we demonstrate \textsc{AlgoFor-DDP-SC}, an approximation algorithm for the \emph{Drone Delivery Packing Problem with Battery Stations and Conflicting Intervals} (DDP-SC). The algorithm works for both swapping and charging stations.
The algorithm uses both the previous algorithm as its subroutine. Later in this section, we present a modified algorithm for the case where the battery service stations operate as swapping stations.
The algorithm is illustrated step by step.%\newline

\textbf{Step-1:} \textit{~Construct the interval graph $G$ from the given set of delivery time intervals. Find $\omega$, the maximum clique size of the interval graph.}%\newline

\textbf{Step-2:} \textit{~Divide the interval set $\cI$ into disjoint subset of intervals $\{\mathcal{I}_1, \mathcal{I}_2, \ldots, \mathcal{I}_{r+1}\}$, where
% $\cI_1 = \{I_j ~|~ t_j^L < t_1^A\}$, and for
% $2 \leq \ell \leq r+1,~ \cI_{\ell} = \{I_j ~|~  I_j\notin \bigcup\limits_{\ell'=1}^{\ell-1} \cI_{\ell'} \text{~and~} (t_j^L < t_{\ell}^A)\}$.
$\cI_1 = \{I_j ~|~ t_j^L < t_1^A\}$, and for 
$2 \leq \ell \leq r, ~ \cI_{\ell} = \{I_j ~|~ t_{\ell -1}^A \leq t_j^L < t_{\ell}^A\}$, and $\cI_{r+1} = \{I_j ~|~ t_j^L \geq t_{\ell}^A\}$.} %\newline

Analogous to the Section \ref{sec:ddp-nc}, we introduce some notions. Unlike DDP-NC, there can be multiple intervals of $\cI_{\ell}$ containing $t_{\ell-1}^D$ and $t_{\ell}^A$. %Furthermore, an interval can be so long that it intersects multiple waiting intervals.
For $\ell \geq 2$, the set of all the intervals in $\cI_{\ell}$ containing $t_{\ell-1}^D$ is denoted by $\cI_{\ell}^{first}$. For $\ell \leq r$, the set of all the intervals in $\cI_{\ell}$ containing $t_{\ell}^A$ is denoted by $\cI_{\ell}^{last}$. We also extend the definition of a \emph{block}, where a block is a compatible set of intervals with total cost at most $B$.
%Now, for each $\ell$, we divide the interval set $\cI_{\ell}$ into a disjoint set of blocks using Algorithm \ref{Alg:approx_ddp_color}, so that we can assign a drone to each of those blocks. Since any assignment returned by Algorithm \ref{Alg:approx_ddp_color} is feasible (Lemma \ref{lem:cor+compl-ddp-ns}) and does not contain any waiting interval, they can be referred to as a block as well.%\newline

\textbf{Step-3:} \textit{~For each $1 \leq \ell \leq r+1$, we divide the interval set $\cI_{\ell}$ into a disjoint set of blocks using Algorithm \ref{Alg:approx_ddp_color}. Let $m_{\ell}$ be the number of blocks returned by this step for the set $\cI_{\ell}$.}
%\newline
Since any assignment returned by Algorithm \ref{Alg:approx_ddp_color} is feasible (Lemma \ref{lem:cor+compl-ddp-ns}) and does not contain any waiting interval, they can be referred to as a block as well.

\textbf{Step-4:} \textit{~Find $m_{max} = \max\limits_{1\le \ell \le r+1}m_{\ell}$. used $m_{max}+2\omega$ drones.}%\newline

We now aim to assign a drone to each of the blocks, starting from the blocks of $\cI_1$.
For $\ell \geq 2$, the set of blocks containing the intervals of $\cI_{\ell}^{first}$
%$\cI_{\ell}^{first} \bigcap \cI_{\ell}$
is denoted by $\mathcal{S}_{\ell}^{first}$. For $\ell \leq r$, the set of blocks containing the intervals of $\cI_{\ell}^{last}$
%$\cI_{\ell}^{last} \bigcap \cI_{\ell}$
is denoted by $\mathcal{S}_{\ell}^{last}$. Moreover, the drones assigned to the intervals of $\cI_{\ell}^{first}$ and $\cI_{\ell}^{last}$ is denoted by $\mathcal{D}_{\ell}^{first}$ and $\mathcal{D}_{\ell}^{last}$, respectively.
%Note that the drones in $\mathcal{D}_{\ell}^{first}$ and $\mathcal{D}_{\ell}^{last}$ may contain the intervals not in $\cI_{\ell}$.
For completeness, we set $\cI_0 = \cD_{-1}^{last} = \cD_0^{last} = \emptyset.$ %\newline

\textbf{Step-5:} \textit{~For $\cI_{\ell}$, we assign the blocks of $\mathcal{S}_{\ell}^{first}$ uniquely to those drones that are not assigned to any of the blocks of $\cI_{\ell-1}$ and not in $\cD_{\ell-2}^{last}$. The remaining blocks are assigned uniquely to the drones except $\cD_{\ell-1}^{last}$ and $\cD_{\ell}^{first}$. After the assignments, if $\ell \le r$, we charge all the used drones (having remaining batteries less than $B$) except those in $\cD_{\ell}^{last}$ for the entire waiting interval $\cI_{\ell}^c$.}%\newline

\noindent We now analyze the algorithm. Lemma \ref{lem:cor-ddp-sc} and \ref{lem:rel-opt-l-opt-sc} are analogous to Lemma \ref{lem:cor-ddp-nc} and \ref{lem:rel-opt-l-opt}, resp.
\begin{lemma}
\label{lem:cor-ddp-sc}
    \textsc{AlgoFor-DDP-SC} assigns each of the deliveries in $\cI$ to one of the $(m_{max}+2\omega)$ drones, and the assignments of all those drones are feasible.
\end{lemma}

\begin{proof}
    From Step-2, an interval $I$ of $\cI$ is part of exactly one $\cI_{\ell}$. Algorithm \ref{Alg:approx_ddp_color} assigns $I$ to a unique assignment, by Lemma \ref{lem:cor+compl-ddp-ns}, and therefore to a unique drone.

    We now prove the feasibility of the assignments by induction on $\ell$, the index of the interval set $(\cI_{\ell})$, where $1 \leq \ell \le r+1$. 
    Initially, the assignments of all the drones are empty.
    For $\cI_1$, algorithm divides the intervals into blocks by using Algorithm \ref{Alg:approx_ddp_color}. By Lemma \ref{lem:cor+compl-ddp-ns}, all those blocks are feasible,
    %, i.e, intervals in each of the blocks are compatible, and the total cost of those intervals is at most $B$.
    and so assigning all intervals of a block to a drone holds the feasibility condition.

    Let us assume that the assignments of all the drones are feasible up to the processing of the intervals of $\cI_{\ell-1}$. 
    Now consider the assignment for the intervals of $\cI_{\ell}$. The algorithm first assigns a block in $\mathcal{S}_{\ell}^{first}$ to those drones $\cD$ that are neither assigned to the blocks of $\cI_{\ell-1}$ nor in $\cD_{\ell-2}^{last}$.
    %Since any arbitrary point may contain in at most $\omega$ many intervals of $\cI$, there are at most $\omega$ many intervals both in $\cI_{j}^{first}$ and $\cI_{j}^{last}$, and so at most $\omega$ many drones both in $\cD_{j}^{first}$ and $\cD_{j}^{last}$, for $j< \ell$.
    Let $I$ be an interval in $\cI_{\ell}^{first}$. Among all the intervals in $\cup_{j=1}^{\ell-1} \cI_j$, the intervals in $\cI_{\ell-1} \bigcup \cI_{\ell-2}^{last}$ are the only intervals that can intersect with $I$. The number of drones assigned to $\cI_{\ell-1}$ is $m_{\ell-1} \leq m_{max}$ and there are at most $\omega$ many drones in the set $\cD_{\ell-2}^{last}$. Therefore, we have at least $\omega$ many drones in $\mathcal{D}$. The assignments of those drones do not have an interval that intersects with the waiting interval $I_{\ell-2}^c$, and so they are of full capacity before the processing of the intervals of $\cI_{\ell}$. Consequently, assigning the blocks of $\mathcal{S}_{\ell}^{first}$ uniquely to the drones in $\cD$, keeping the feasibility condition satisfied.
    All the remaining blocks of $\cI_{\ell}$ contain the intervals with the launch time greater than $t_{\ell-1}^D$.
   All the drones except that in $\cD_{\ell-1}^{last} \bigcup \cD_{\ell}^{first}$ do not intersect with the waiting interval $I_{\ell-1}^c$, ans so they are of full capacity at $t_{\ell-1}^D$, after recharging at the station $(\ell-1)$.
   Since, number of drones in $\cD_{\ell-1}^{last} \bigcup \cD_{\ell}^{first}$ is at most $2\omega$, there are at least $m_{max}$ drones, which is of full capacity at $t_{\ell-1}^D$, and thus can be uniquely assigned to the remaining at most $m_{max}$ blocks of $\cI_{\ell}$ without violating the feasibility of the assignments. Thus, the assignments of all the drones remain feasible after the processing of the intervals in $\cI_{\ell}$.
\end{proof}

\begin{lemma}
    \label{lem:rel-opt-l-opt-sc}
    For any $\ell : 1 \le \ell \le r+1$, $OPT_{\ell} \leq OPT_{SC}$, where $OPT_{\ell}$ is the number of blocks in an optimum partition of $\cI_{\ell}$, and $OPT_{NC}$ is the optimum solution for DDP-SC.
\end{lemma}

\begin{proof}
    This can be proved via similar arguments as in the proof of Lemma \ref{lem:rel-opt-l-opt}. A drone $D$ in any solution of DDP-SC can not have an assignment that contains some intervals of $\cI_{\ell}$, then recharge at the station $(\ell-1)$, and then again to some of the intervals of $\cI_{\ell}$. Hence, the cost of the intervals of $\cI_{\ell}$ assigned to $D$ can be at most $B$, and so those intervals constitute a block for $\cI_{\ell}$. 

    Therefore, we can always find a partition of $\cI_{\ell}$ from the optimum solution of DDP-SC, which is of size at most $OPT_{SC}$. Since $OPT_{\ell}$ is the size of the optimum partition of $\cI_{\ell}$, we have $OPT_{\ell} \le OPT_{SC}$. 
\end{proof}

\begin{theorem}
\label{thm:alg-ddp-sc}
   \textsc{AlgoFor-DDP-SC} is an approximation algorithm for DDP-SC that uses at most ($\beta \cdot OPT_{NC} +2 \omega$) drones, where $\beta$ is the approximation factor for Algorithm \ref{Alg:approx_ddp_color}.
\end{theorem}

\begin{proof}
    The feasibility of the algorithm follows from Lemma \ref{lem:cor-ddp-sc}. The running time of the algorithm dominated by the subroutine used in Step-3 of the algorithm, which takes $O(n \log n + n_e + r)$ time, by Lemma \ref{lem:cor+compl-ddp-ns} and Theorem \ref{thm:alg-ddp-nc}. 

    Let $\ell^* = \arg \max\limits_{1 \le \ell \le r+1} m_{\ell}$. Then the number of used drones by the algorithm is $m_{max} + 2\omega = m_{\ell^*}+2\omega$, which is at most $\beta \cdot OPT_{\ell^*} + 2\omega \le  \beta \cdot OPT_{SC} + 2\omega, \text{~by Lemma \ref{lem:rel-opt-l-opt-sc}.}$
    %Hence the proof.
\end{proof}

\begin{discussion}
    From the above analysis and using Corollary \ref{cor:genralized_factor}, we can infer that our algorithm is an $(4+\psi)$-approximation algorithm for DDP-SC, where $\psi = \frac{\epsilon_{max} - \epsilon_{min}}{1 - \epsilon_{max}}$ and $\epsilon_{min} = \frac{1}{B}\min_{j \in \cN} cost(I_j)$ and $\epsilon_{max} = \min\{\frac{1}{2}, \frac{1}{B}\max_{j \in \cN}cost(I_j)\}$. In the worst case, the algorithm uses at most $\big((2+\frac{\epsilon_{max} - \epsilon_{min}}{1 - \epsilon_{max}})OPT_{SC} + 2\omega \big)$ drones. In the next section, we modify our algorithms that open at most $(3+\psi)OPT_{SC}$ drones, when the battery stations are the swapping stations.
\end{discussion}

\begin{discussion}
\label{dis:motivation-for-modification}
    If we use any algorithm for DDP-NS instead of Algorithm~\ref{Alg:approx_ddp_color} in Step~3, then $\beta$ is replaced by the corresponding approximation factor. Therefore, if we use the result of Epstein and Levin for BPPC on interval graphs~\cite{epstein} instead of Algorithm~\ref{Alg:approx_ddp_color}, \textsc{AlgoFor-DDP-SC} uses at most $(\frac{7}{3}OPT_{SC} + 2\omega)$ drones. In other words, in this case, \textsc{AlgoFor-DDP-SC} becomes a $(4+\frac{1}{3})$-approximation algorithm for DDP-SC (using $\omega \leq OPT_{SC}$). Even in this case, our subsequent modified algorithm for DDP-SC achieves a better approximation ratio of $(3+\psi)$, which is at most $4$ in the worst case.
\end{discussion}

\subsection{Modification to \textsc{AlgoFor-DDP-SC}}
\label{sec:mod-ddp-sc}
Here, we aim to modify the algorithm when the battery station is the swapping station. We illustrate our algorithm with an example in %Section 3 of the supplemental file.
Section \ref{sec:app-mod-ddp-sc}, Appendix \ref{sec:appendix}. 
The modified algorithm is referred to as \textsc{Mod-Alg-DDP-SC}. In this algorithm, we use one of the properties of the swapping station: if a drone decides to swap its battery at the $\ell$-th station, then it can not be assigned to any delivery interval that intersects the corresponding swapping interval $I_{\ell}^c$. We also use Algorithm \ref{Alg:approx_ddp_color} as the subroutines and while using this we exploit one of its properties as mentioned in Discussion \ref{dis:approx-ddp-color-2}. We now present the modified algorithm, where the $i$-th step is denoted Step$-i^+$.

 \textbf{Step-1$^+$:} \textit{~Construct the interval graph $G$ from the given set of delivery time intervals. Find $\omega$, the maximum clique size of the interval graph.}%\newline

 \textbf{Step-2$^+$:} \textit{~Divide the interval set $\cI$ into disjoint subset of intervals $\{\widetilde{\mathcal{I}}_1, \widetilde{\mathcal{I}}_2, \ldots, \widetilde{\mathcal{I}}_{r+1}\}$, where
% $\cI_1 = \{I_j ~|~ t_j^L < t_1^A\}$, and for
% $2 \leq \ell \leq r+1,~ \cI_{\ell} = \{I_j ~|~  I_j\notin \bigcup\limits_{\ell'=1}^{\ell-1} \cI_{\ell'} \text{~and~} (t_j^L < t_{\ell}^A)\}$.
$\widetilde{\cI}_1 = \{I_j ~|~ t_j^L \le t_1^D\}$, and for 
$2 \leq \ell \leq r, ~ \widetilde{\cI}_{\ell} = \{I_j ~|~ t_{\ell -1}^D < t_j^L \le t_{\ell}^D\}$, and $\widetilde{\cI}_{r+1} = \{I_j ~|~ t_j^L > t_{\ell}^D\}$.} %\newline

For each $\ell$ $(1 \le \ell \le r+1)$, we denote $G_{\ell}$ to refer the subgraph of $G$ induced by the corresponding vertices of $\widetilde{\cI}_{\ell}$.
For each $\ell$ $(1 \le \ell \le r)$, the set of intervals of $\widetilde{\cI}_{\ell}$ containing the time $t_{\ell}^A$ is denoted by $\cI_{\ell}^{left}$, and the set of intervals of $\widetilde{\cI}_{\ell}$ containing the time $t_{\ell}^D$ but not in $\cI_{\ell}^{left}$ is denoted by $\cI_{\ell}^{right}$.%\newline

 \textbf{Step-3$^+$:} \textit{~For each $\ell$  $(1 \le \ell \le r)$, we construct the bipartite graph $G_{\ell}^b (V_{\ell}^{left} \bigcup V_{\ell}^{right}, ~ E_{\ell}^b$), where $V_{\ell}^{left}$ and $V_{\ell}^{right}$ are the set of vertices corresponding to $\cI_{\ell}^{left}$ and intervals of $\cI_{\ell}^{right}$, respectively. Two vertices $u \in V_{\ell}^{left}$ and $v \in V_{\ell}^{right}$ are adjacent in $E_{\ell}^b$ if the corresponding intervals are not in conflict and the sum of their costs do not exceed the budget $B$.}%\newline

Note that $G_{\ell}^b$ is a subgraph of the complement graph of $G_{\ell}$. We now aim to find a maximum matching in each bipartite graph. A \emph{matching} is a set of independent edges, and \emph{maximum matching} is the matching with maximum cardinality.%\newline

\noindent \textbf{Step-4$^+$:} \textit{~For each $\ell$  $(1 \le \ell \le r)$, find the maximum matching $\mathcal{M}_{\ell}$ of $G_{\ell}^b$. Let $x_{\ell}$ be the size of the matching and let $z_{\ell}$ be the sum of $x_{\ell}$ and all unmatched vertices of $G_{\ell}^b$. In other words, $z_{\ell} = |V_{\ell}^{left} \bigcup V_{\ell}^{right}| - x_{\ell}$.} %\newline

 \textbf{Step-5$^+$:} \textit{~For each $\ell$  $(1 \le \ell \le r)$, color the vertices of $G_{\ell}$ using $\max\{\omega, z_{\ell}\}$ colors such that no two adjacent vertices receive the same color. First, we assign the color to the vertices of $G_{\ell}^b$. If $(u,v)$ be a matched vertex in $\mathcal{M}_{\ell}$, we assign the same color for both the vertices. After assigning the color to all the matched vertices, the unmatched vertices of $G_{\ell}^b$ are colored by using $(z_{\ell}-x_{\ell})$ additional distinct colors. 
Thereafter, the remaining vertices of $G_{\ell}$ are colored sequentially in non-increasing order of the rendezvous times of the corresponding intervals. 
Specifically, when coloring a vertex $u$ according to this order, assign to $u$ any color that is not used by its already colored neighbors.
Furthermore, the vertices of $G_{r+1}$ are colored by $\omega$ colors using standard interval coloring procedure.}% without any restriction.}%\newline

\textbf{Step-6$^+$:} \textit{After the coloring, for each $\ell$  $(1 \le \ell \le r+1)$, we apply Algorithm \ref{Alg:approx_ddp_color} to partition the intervals corresponding to the vertices of $G_{\ell}$ into distinct blocks. The only restriction is that, while processing the intervals with color $k (1 \leq k \leq z_{\ell})$ in Step-6, Algorithm \ref{Alg:approx_ddp_color}, we first assign a drone to the intervals (if exists) corresponding to the matched vertices of $M_{\ell}$ so that they can be part of the same block. Let $\widetilde{m}_{\ell}$ be the number of blocks returned by the algorithm.}%\newline

The set of intervals corresponding to the vertices of $V_{\ell}^{left}$ and $V_{\ell}^{right}$ is denoted by $\cI_{\ell}^{left}$ and $\cI_{\ell}^{right}$, respectively.
%Moreover, note that an interval can be so long that it intersects multiple waiting intervals. For $1 < \ell < r$, the set of intervals in $\cI$ that does not in $\cI_{\ell}$ but intersect with the waiting interval $I_{\ell}^c$ is named as $\cI_{\ell}^{out}$, and the size of the set is denoted by $y_{\ell}$.
%\newline

 \textbf{Step-7$^+$:} \textit{~Find $z_{max} = \max\limits_{1 \le \ell \le r+1} z_{\ell}$ and $\widetilde{m}_{max} = \max\limits_{1 \le \ell \le r+1} \widetilde{m}_{\ell}$. used $(\widetilde{m}_{max}+z_{max})$ drones.}%\newline

For $\ell \le r$, we denote the set of intervals in $\widetilde{\cI}_{\ell}^{left} \bigcup \cI_{\ell}^{right}$ as $\cI_{\ell}^{ext}$ and the set of drones assigned for these intervals is named as $\cD_{\ell}^{ext}$. For the completeness, we assume $\cD_{\ell}^{ext} = \emptyset$. %\newline

\textbf{Step-8$^+$:} \textit{~For each $\ell$  $(1 \le \ell \le r+1)$, we assign the blocks of $\widetilde{\cI}_{\ell}$ uniquely to those drones that are not in $\cD_{\ell-1}^{ext}$. After the assignments, if $\ell \le r$, we swap the battery of all the used drones except those that are in $\cD_{\ell}^{ext}$ at the $\ell$-th station.}

In the following, we analyze our algorithm:
\begin{lemma}
    \label{lem:proper-col-ddp-sc-mod}
    For each $\ell$  $(1 \le \ell \le r+1)$, the Step$-3^+$ of \textsc{Mod-Alg-DDP-SC} color the vertices of $G_{\ell}$ properly, i.e., no two vertices of $G_{\ell}$ get the same color.
\end{lemma}
\begin{proof}
     Since $\omega$ is the maximum clique size of the interval graph $G$, the maximum clique size of any of the subgraph $G_{\ell}$ is at most $\omega$. If $\ell = r+1$, then the vertices of $G_{r+1}$ are colored without any restriction, and hence are properly colored using $\omega$ colors.

    Now consider $\ell$ with $1 \le \ell \le r$. In this case, some vertices are colored separately as described in Step~$3^+$. 
    %First, the algorithm colors the vertices corresponding to the matched edges.
    If $(u,v)$ is an edge in the matching $\mathcal{M}_{\ell}$, then by the construction of $G_{\ell}^b$, the corresponding intervals do not conflict, and hence assigning the same color to both vertices preserves the coloring property.
    Thereafter, assigning distinct colors to the remaining $(z_{\ell}-x_{\ell})$ unmatched vertices in $G_{\ell}^b$ trivially satisfies the coloring property.

    After coloring the vertices of $G_{\ell}^b$, the remaining vertices of $G_{\ell}$ are colored using the standard greedy interval-coloring procedure. The vertices are processed in non-increasing order of the rendezvous times of their corresponding intervals. 
    All intervals in this step have rendezvous time at most $t_{\ell}^A$, while the intervals corresponding to vertices in $G_{\ell}^b$ have rendezvous time strictly greater than $t_{\ell}^A$. Let $v_1$ be the uncolored vertex with the largest rendezvous time. All of its already colored neighbors correspond to intervals whose rendezvous times are at least that of $v_1$. Since the clique number of the graph is $\omega$, there are at most $(\omega-1)$ such neighbors. Hence there exists a color not used by these neighbors, which can be assigned to $v_1$.
    More generally, when processing a vertex $v_j$ in this order, all previously colored neighbors of $v_j$ correspond to intervals whose rendezvous times are at least that of the interval corresponding to $v_j$. 
    There are at most $(\omega-1)$ such neighbors, and so at least one color remains available for $v_j$.
    So all vertices of $G_{\ell}$ can be properly colored using $\max\{z_{\ell},\omega\}$ colors.
\end{proof}

\begin{lemma}
    \label{lem:num-drones-in-D-l-ext}
    The number of drones in $\mathcal{D}_{\ell}^{ext}$ is $z_{\ell}$.
\end{lemma}
\begin{proof}
    Let $(u,v)$ be a matching edge of $\mathcal{M}_{\ell}$ in $G_{\ell}^b$. The corresponding intervals get the same color. According to Step $6^+$, the algorithm first processes the two intervals associated with the vertices $u$ and $v$. By doing so, they become part of the same block, as from the construction of $G_{\ell}^b$, the sum of their costs is less than $B$. Thus, both intervals are assigned to the same drone. Since, $x_{\ell}$ is the number of matching edges in $\mathcal{M}_{\ell}$, we need $x_{\ell}$ drones for all the matched vertices (equiv. intervals). There are $(z_{\ell} - x_{\ell})$ many unmatched vertices and they get the distinct color, and so assign to $(z_{\ell}-x_{\ell})$ distinct drones.
\end{proof}

\begin{lemma}
    \label{lem:feasible-ddp-sc-mod}
    \textsc{Mod-Alg-DDP-SC} assigns each of the deliveries in $\cI$ to one of the $(\widetilde{m}_{max}+ z_{max})$ drones, and the assignments of all those drones are feasible.
\end{lemma}
\begin{proof}
    We prove the lemma by induction on $\ell$, where $\ell$ is the index of the $\ell$-th interval set ($\widetilde{\cI}_{\ell}$).
    For $\ell=1$, the $\widetilde{m}_1 (\le \widetilde{m}_{max})$ blocks of $\widetilde{\cI}_1$ are uniquely assigned to any of the $m_1$ drones. All of these drones initially have empty assignments, so after the assignment, the feasibility property is trivially satisfied.

   From the induction hypothesis, we assume that the lemma holds up to the processing of the intervals of $\widetilde{\cI}_{\ell-1}$. For $\widetilde{\cI}_{\ell}$, the intervals of $\cI_{\ell-1}^{ext}$ are the only intervals among all the currently processed intervals that can be in conflict with the intervals of $\cI_{\ell}$. All the assignments of the used drones except those that are in $\cD_{\ell-1}^{ext}$ do not contain an interval that intersects with the waiting interval $\widetilde{\cI}_{\ell-1}^c$. All these drones can swap their battery at the swapping station $(\ell-1)$, and their battery capacity becomes full before the processing of the intervals of $\widetilde{\cI}_{\ell}$. Since there are at most $z_{max}$ many drones in $\cD_{\ell}^{ext}$, there are at least $\widetilde{m}_{max}$ drones which are of full capacity before processing $\widetilde{\cI}_{\ell}$. Hence, assigning all the $\widetilde{m}_{\ell} (\le \widetilde{m}_{max})$ blocks uniquely to those of $\widetilde{m}_{max}$ drones keeps the feasibility condition satisfied.
\end{proof}

\begin{lemma}
    \label{lem:rel-opt-l-opt-ddp-sc-mod}
    For each $\ell$  $(1 \le \ell \le r+1)$, $\widetilde{OPT}_{\ell} \leq OPT_{SC}$, where $\widetilde{OPT}_{\ell}$ is the optimum number of blocks in the optimum partition of the interval set $\widetilde{\cI}_{\ell}$.
\end{lemma}
\begin{proof}
   Any feasible assignment in DDP-SC cannot contain a non-empty set of intervals from $\widetilde{\mathcal{I}}_{\ell}$, then swap its battery at station $\ell$, and subsequently contain another non-empty set of intervals from $\widetilde{\mathcal{I}}_{\ell}$. Therefore, the total cost of all intervals from $\widetilde{\mathcal{I}}_{\ell}$ assigned to a drone is at most $B$.
    Consequently, if we extract the intervals of $\widetilde{\mathcal{I}}_{\ell}$ from each assignment in an optimal solution of DDP-SC, each such set forms a block of $\widetilde{\mathcal{I}}_{\ell}$. In other word $OPT_{SC}$ serves as the upper bound on the optimum partition of $\widetilde{\cI}_{\ell}$. Hence, $\widetilde{OPT}_{\ell} \le OPT_{SC}$.
\end{proof}

\begin{lemma}
    \label{lem:rel-z-l-opt-ddp-sc-mod}
    For each $\ell$ with $1 \le \ell \le r$, $z_{\ell} \leq OPT_{SC}$.
\end{lemma}
\begin{proof}
    On the contrary, let $z_{\ell} > OPT_{SC}$. Let $z_{\ell} = x_{\ell} + y_{\ell}$, where $x_{\ell}$ is size of the maximum matching $\mathcal{M}_{\ell}$ of the bipartite graph $G_{\ell}^b$ and $y_{\ell}$ is the number of unmatched vertices of $G_{\ell}^b$. 

    Since the intervals in $\cI_{\ell}^{first}$ contain the time $t_{\ell}^A$, a feasible assignment can not have more than one interval of $\cI_{\ell}^{left}$. Similarly, a feasible assignment can not have more than one interval of $\cI_{\ell}^{right}$. Let $OPT_{SC}'$ be the number of assignments in the optimum solutions of DDP-SC that contain intervals from $\cI_{\ell}^{ext}$. Among them, let $x_{\ell}^*$ be the number of assignments that contain intervals one from each $\cI_{\ell}^{left}$ and $\cI_{\ell}^{right}$. Then, corresponding vertices create an edge of $G_{\ell}^b$, and all such edges form a matching of $G_{\ell}^b$. Now $x_{\ell}$ being the size of the maximum matching of $G_{\ell}^b$, $x_{\ell} \geq x_{\ell}^*$ must holds.
    % the following must holds
    % \begin{equation}
    %     \label{eq:rel-x-l-x*}
    %     x_{\ell} \geq x_{\ell}^*
    % \end{equation}
    
    Let $y_{\ell}^*$ be the number of assignments in the optimum solution containing intervals only from $\cI_{\ell}^{left}$ or $\cI_{\ell}^{right}$. Then, $OPT_{SC}' = x_{\ell}^* + y_{\ell}^*$. Moreover, 
    $|\cI_{\ell}^{left} \bigcup \cI_{\ell}^{right}| = 2x_{\ell}+y_{\ell} = 2x^*_{\ell} + y_{\ell}^*$.

    From our assumption, $z_{\ell} > OPT_{SC}$, implies $z_{\ell} > OPT_{SC}'$ as well. Therefore,
    \vspace{-1em}
    \begin{align*}
        x_{\ell} + y_{\ell} > x_{\ell}^* + y_{\ell}^*
        \implies& 2x_{\ell} + y_{\ell} > x_{\ell} + x_{\ell}^* + y_{\ell}^*
        \implies & 2x_{\ell}^* + y_{\ell}^* > x_{\ell} + x_{\ell}^* + y_{\ell}^*
        \implies & x_{\ell}^* > x_{\ell}
    \end{align*}
    \vspace{-1em}
    which contradicts the inequality $x_{\ell} \geq x_{\ell}^*$
    %\ref{eq:rel-x-l-x*}. 
    Hence, $z_{\ell} \leq OPT_{SC}$.
\end{proof}

\begin{theorem}
    \label{thm:final-ddp-sc-mod}
    \textsc{Mod-Alg-DDP-SC} is an $(3+ \frac{\epsilon_{max} - \epsilon_{min}}{1 - \epsilon_{max}})$-approximation algorithm for DDP-SC.
    %, where $\psi = \frac{\epsilon_{max} - \epsilon_{min}}{1 - \epsilon_{max}}$.
\end{theorem}
\begin{proof}
    The feasibility of the algorithm follows from Lemma \ref{lem:feasible-ddp-sc-mod}.
    The running time of the algorithm is dominated by Step-$3^+ - 6^+$. Constructing the bipartite graph and then finding the maximum matching takes $O(n^{2.273})$ time using a faster matrix multiplication method \cite{mucha2004maximum, alman2024refined}. After that, finding the block for each interval set takes $O(n\log n + n_e + r)$ time in total, from Lemma \ref{lem:cor+compl-ddp-ns} and Theorem \ref{thm:alg-ddp-nc}. All the other steps of the algorithm require linear time, and so the algorithm runs in $O(n^{2.273}+r)$ time.

    The number of used drones by the algorithm is ($\widetilde{m}_{max}+z_{max}$). Algorithm finds $\widetilde{m}_{\ell}$ using Algorithm \ref{Alg:approx_ddp_color}. However it uses $\max\{z_{\ell}, \omega\}$-coloring instead of $\omega$-coloring. Therefore, from Eq. \ref{eq-14}, $\widetilde{m}_{\ell} \leq \frac{\widetilde{OPT}_{\ell}}{1-\epsilon_{max}} + \max\{z_{\ell}, \omega\} (1 - \frac{\epsilon_{min}}{1-\epsilon_{max}})$.
    %Let $\ell^* = \arg \limits_{1\le \ell \le r+1} (\widetilde{m}_{\ell}+z_{\ell})$.
    So, by Lemma \ref{lem:rel-opt-l-opt-ddp-sc-mod} and \ref{lem:rel-z-l-opt-ddp-sc-mod}, 
    \vspace{-1em}
    %\footnotesize{
    \begin{align*}
        \widetilde{m}_{max}+z_{max} \leq & ~\frac{\widetilde{OPT}_{\ell^*}}{1-\epsilon_{max}} + \max\{z_{\ell^*}, \omega\}(1 - \frac{\epsilon_{min}}{1-\epsilon_{max}}) + z_{\ell}^*, \text{~where }\ell^* = \arg \limits_{1\le \ell \le r+1} (\widetilde{m}_{\ell}+z_{\ell}) \\
        \leq & ~\frac{OPT_{SC}}{1-\epsilon_{max}} + OPT_{SC}(1 - \frac{\epsilon_{min}}{1-\epsilon_{max}}) + OPT_{SC}, %\text{~by Lemma \ref{lem:rel-opt-l-opt-ddp-sc-mod}, \ref{lem:rel-z-l-opt-ddp-sc-mod}}
        \\
        \leq & ~(3+\frac{\epsilon_{max} - \epsilon_{min}}{1 - \epsilon_{max}})~OPT_{SC}
    \end{align*}%}
    \vspace{-1em}
    Hence the proof.
\end{proof}
\section{Performance Evaluation}
\label{sec:experiment}
In this section, we validate our proposed algorithms and compare them with the optimum solution in terms of both the optimization value and time. We assume that the battery stations are the swapping stations.

\subsection{The Settings}
In our experiment, we generate several random intervals from a fixed distribution. 
%We consider the endpoints of the intervals are integers.
We assume that the duration of the truck in a single day is $0$ to $T = 300$ units of time, which we have scaled from the practical duration trip of $30000 s$ \cite{campbell2017strategic, Betti}.
The endpoints of the intervals are integers, where the starting times lie in $[0, 300)$.
We consider the number of delivery intervals $n = 50, 70, 100, 150, 200$. For the instance of DDP-SC, we generate intervals $\{20, 30, \cdots, 80\}$. The smaller instance is taken due to the higher number of constraints of DDP-SC.

% \begin{figure}[htbp]
% \centering
% \includegraphics[width=1\columnwidth]{ns-uniform-20_cropped.pdf}
% \includegraphics[width=\columnwidth]{ns-uniform-50_cropped.pdf}
% \includegraphics[width=\columnwidth]{ns-exp-20_cropped.pdf}
% \includegraphics[width=\columnwidth]{ns-exp-50_cropped.pdf}
% \caption{Performance Evaluation for DDP-NS.}
% \label{fig:ddp-ns}
% \end{figure}

Since the arrival times of requests from any customer are independent and random, we generate random arrival times uniformly in $[0, 300)$ \cite{Betti}.
To do so, we first generate inter-arrival times from an exponential distribution with mean $T/n$. After that, we take the cumulative sum to obtain the actual arrival times of the intervals, which essentially follow a standard Poisson process with rate $n/T$. 
Finally, we scale it by dividing the intervals by the last arrival time, and then multiplying by $(T-1)$, which gives us the final set of arrival times.
Since the delivery cost is assumed to be a linear function of the interval length, we take the cost of each interval to be equal to its length. 
We assume the drone battery budget $(B)$ to be either $5000$ or $2000$~kJ, following \cite{stolaroff2018energy, Betti}; for our experiments, these values are scaled down to $50$ and $20$, respectively.
We consider the distribution length to follow either an exponential or a uniform distribution. The $exponential(B/2)$ distribution yields many small values and a few large ones, so we cover the possible extreme cases. Whereas, by taking the lengths from the uniform distribution $uniform(1, 10)$, we restrict the case to a single delivery interval that cannot be too long, and the positions of the customers are nearly uniformly distributed along the truck's path.
For DDP-NC and DDP-SC, we consider $r = 3$ and $5$ battery stations. The positions of the stations are placed nearly uniformly along the truck route, with random noise in the range $(-2, 2)$. Each battery swapping interval has a length of 5.

For each combination of fixed parameter values $n$, $B$, $r$, and the distribution type of the lengths, we generate five data sets and report the average results.
We compare our algorithm to the optimal solution, which is obtained via ILP as described in Section 1 of the supplemental file.
We use $OPT\text{-}NS$, $OPT\text{-}NC$, and $OPT\text{-}SC$ to denote the average optimal values for DDP-NS, DDP-NC, and DDP-SC, respectively. Similarly, $ALGO\text{-}NS$, $ALGO\text{-}NC$, and $ALGO\text{-}SC$ represent the average number of drones returned by our proposed algorithm for these problems.
The modified algorithm introduced in Section \ref{sec:mod-ddp-sc} is denoted by $M\text{-}ALGO$.
Moreover, $OPT\text{-}Time$ and $ALG\text{-}Time$ represent the average time (in milliseconds) required to solve the problem using the optimal method and our proposed algorithm, respectively.
Additionally, $Clic$ denotes the average clique size of the interval graph constructed from the generated delivery intervals. The notation $c\,OPT\text{-}SC$ refers to the graph obtained from $OPT\text{-}SC$ by scaling all associated values by a constant factor $c$.

In all our figures, the $x$-axis indicates the number of deliveries. The left $y$-axis shows the number of drones, and the right $y$-axis shows the time (in milliseconds).

\begin{figure}[htbp]
\centering
\begin{minipage}{0.48\columnwidth}
    \centering
    \includegraphics[width=\linewidth]{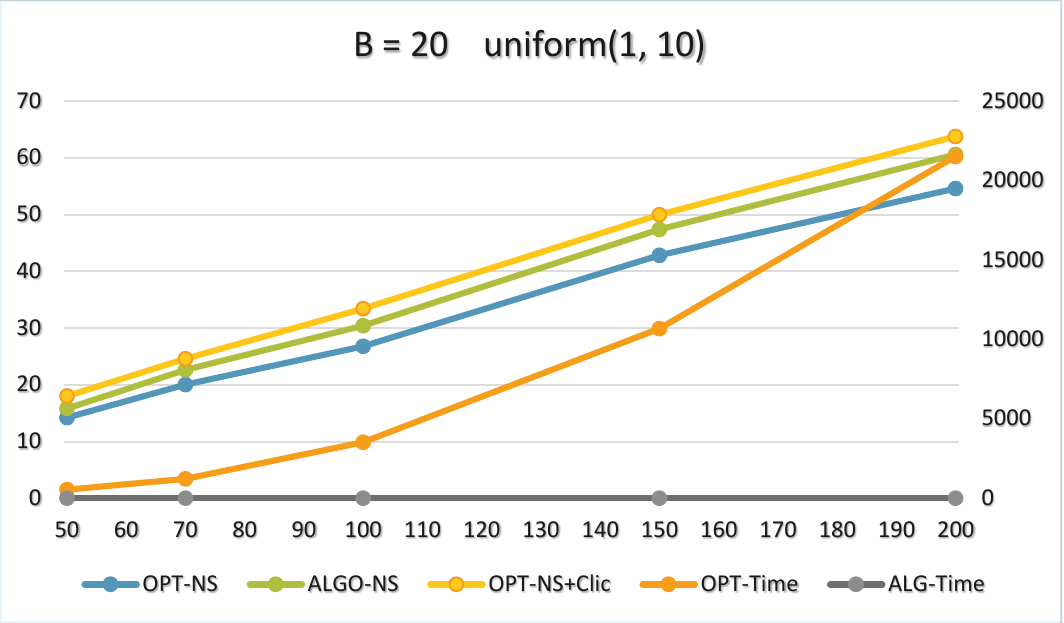}
\end{minipage}
\hfill
\begin{minipage}{0.48\columnwidth}
    \centering
    \includegraphics[width=\linewidth]{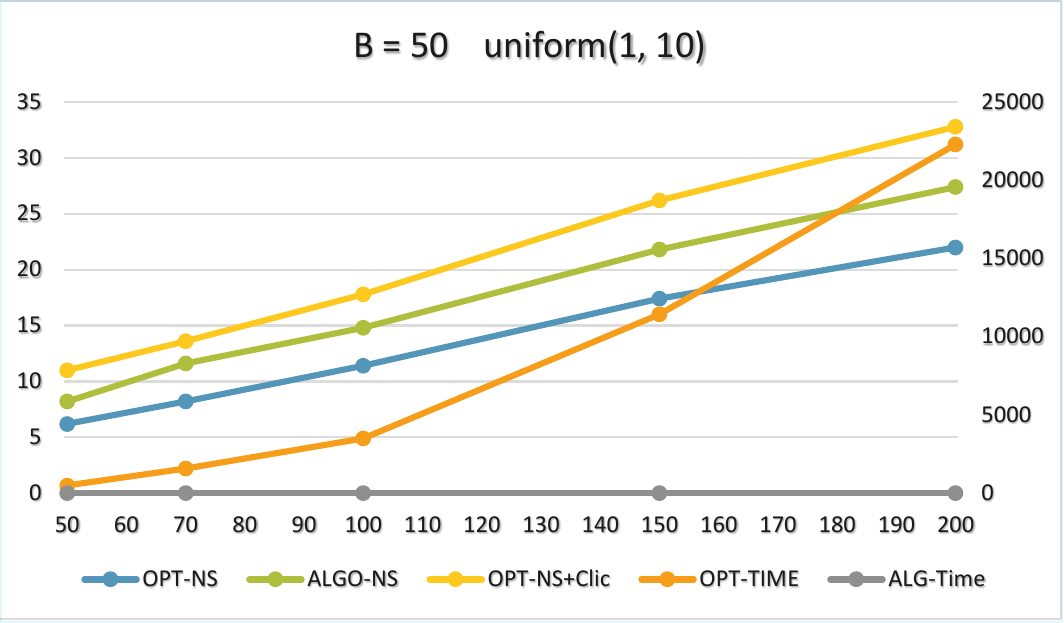}
\end{minipage}
\end{figure}
\begin{figure}[htbp]
\centering
\begin{minipage}{0.48\columnwidth}
    \centering
    \includegraphics[width=\linewidth]{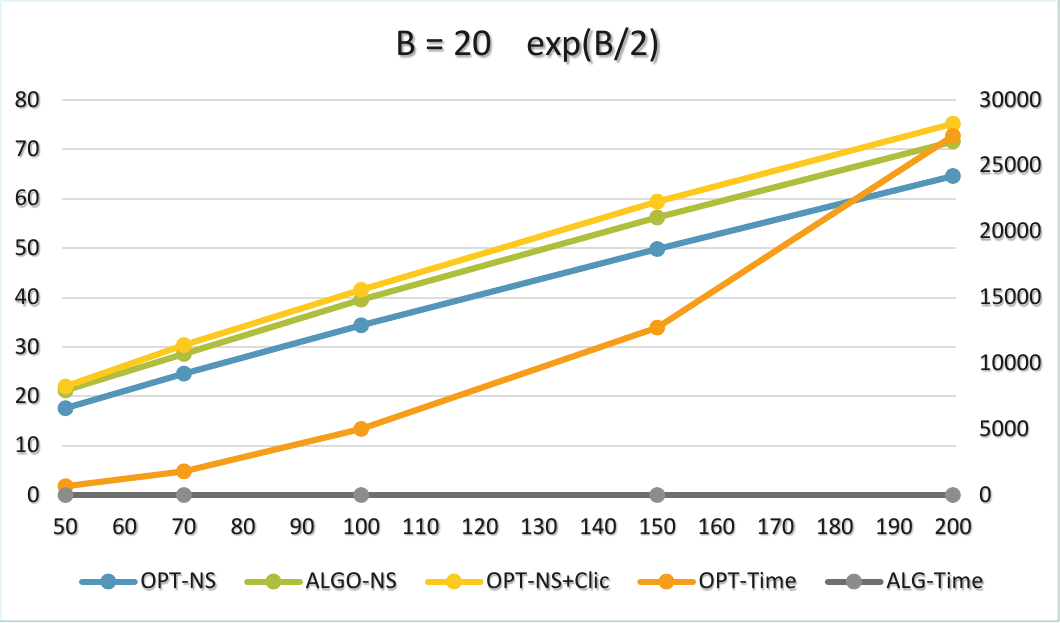}
\end{minipage}
\hfill
\begin{minipage}{0.48\columnwidth}
    \centering
    \includegraphics[width=\linewidth]{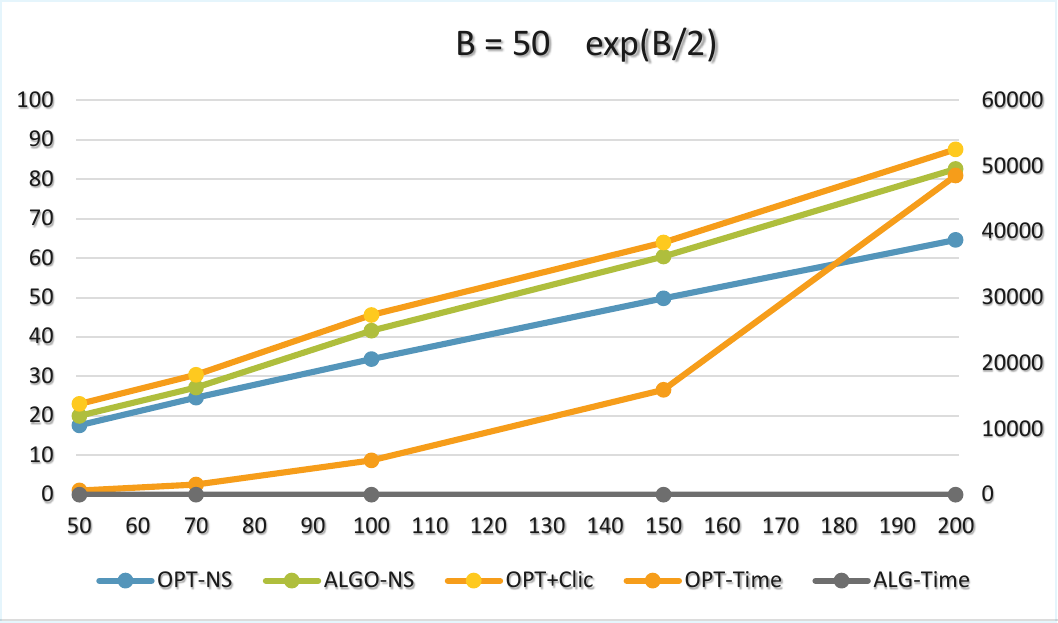}
\end{minipage}
\caption{Performance Evaluation for DDP-NS.}
\label{fig:ddp-ns}
\end{figure}
\subsection{Experimental results}
Fig.~\ref{fig:ddp-ns} presents a performance comparison between our proposed algorithm \textsc{AlgoFor-DDP-NS} and the optimal solution for DDP-NS. We consider four settings with $B \in \{20,50\}$, where the interval lengths follow either a uniform or an exponential distribution. 
Across all instances, we observe that the solution produced by our algorithm is at most the sum of the optimal value and the clique number of the interval graph constructed from the generated intervals. Consequently, our solution is at most twice the optimal value, which is consistent with the theoretical upper bound established in Corollary~\ref{cor:genralized_factor}.
Furthermore, we observe that the time taken by the optimal solution (obtained by solving an ILP) grows exponentially with the number of deliveries, whereas our algorithm takes less than 1 millisecond per instance.

% \begin{figure}[htbp]
% \centering
% \includegraphics[width=\columnwidth]{nc-uniform-3_cropped.pdf}
% \includegraphics[width=\columnwidth]{nc-unifrom-5_cropped.pdf}
% \includegraphics[width=\columnwidth]{nc-exp-3_cropped.pdf}
% \includegraphics[width=\columnwidth]{nc-exp-5_cropped.pdf}
% \caption{Performance Evaluation for DDP-NC.}
% \label{fig:ddp-nc}
% \end{figure}

Fig.~\ref{fig:ddp-nc} presents a performance comparison between our proposed algorithm \textsc{AlgoFor-DDP-NC} and the optimal solution for DDP-NC. We consider four settings with $r \in \{3,5\}$, where the interval lengths follow either a uniform or an exponential distribution. Here we fix the battery budget $B$ as $50$. 
Across all instances, we observe that the solution produced by our algorithm is at most the sum of the optimal value and two.
The empirical result shows that the practical bounds of our proposed algorithm for our chosen instances behave better compared to the shown theoretical bounds on Theorem \ref{thm:alg-ddp-nc} and Discussion \ref{disc:ddp-nc-1}.
This occurs possibly due to our chosen distribution, which is unable to consider the worst instance that matches the theoretical bound, as discussed in Discussion \ref{disc:ddp-nc-1}.
As with DDP-NS, we observe that the optimal solution's running time grows exponentially with the number of deliveries, while our algorithm requires less than 2 milliseconds per instance.
\begin{figure}[htbp]
\centering
\begin{minipage}{0.48\columnwidth}
    \centering
    \includegraphics[width=\linewidth]{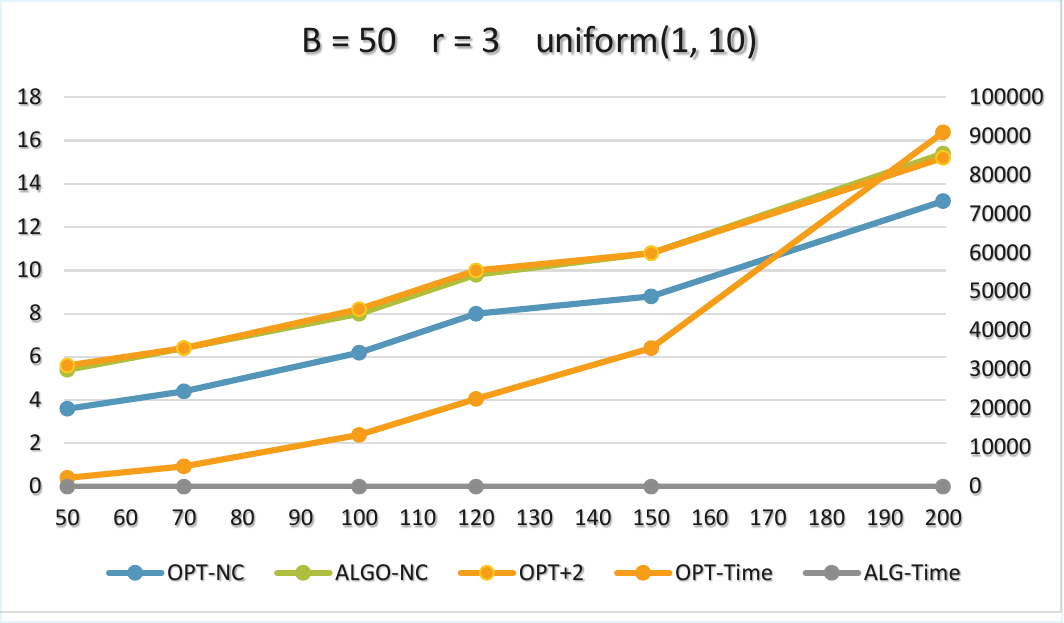}
\end{minipage}
\hfill
\begin{minipage}{0.48\columnwidth}
    \centering
    \includegraphics[width=\linewidth]{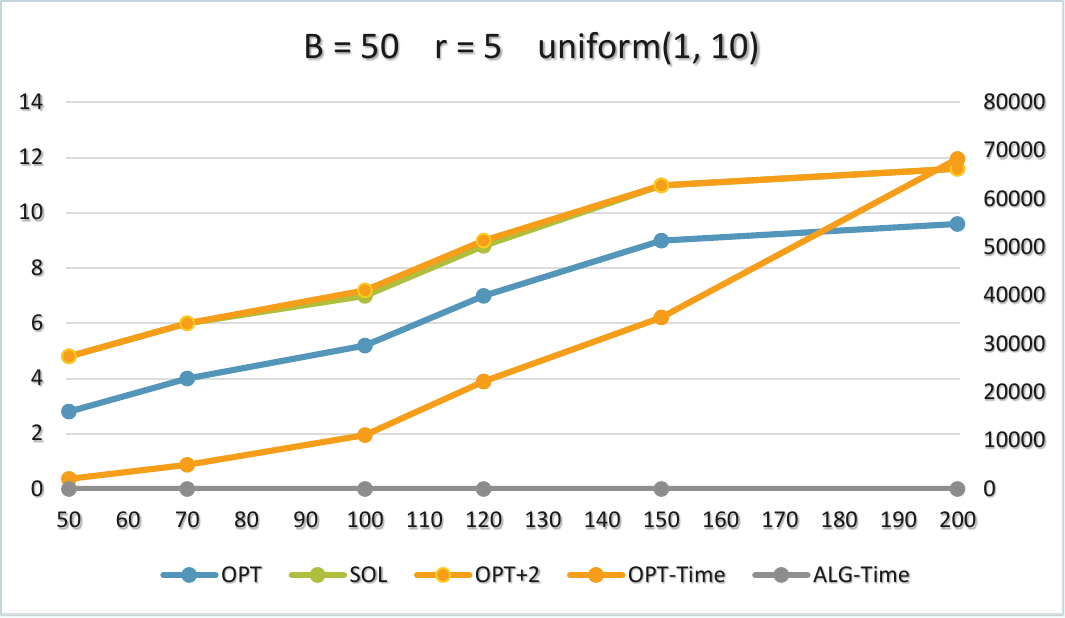}
\end{minipage}
\end{figure}
\begin{figure}[htbp]
\centering
\begin{minipage}{0.48\columnwidth}
    \centering
    \includegraphics[width=\linewidth]{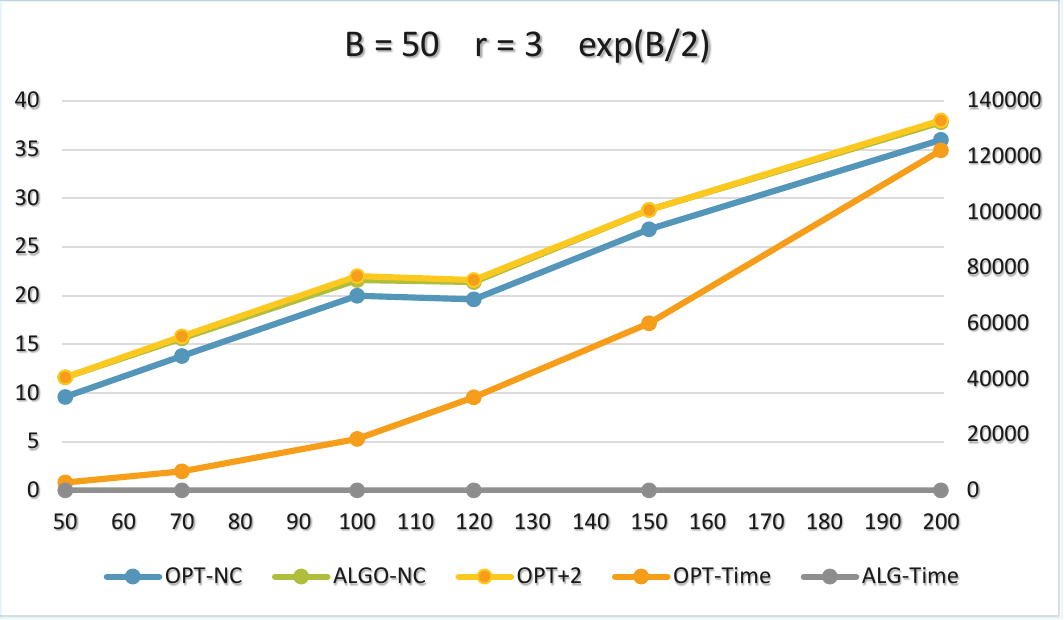}
\end{minipage}
\hfill
\begin{minipage}{0.48\columnwidth}
    \centering
    \includegraphics[width=\linewidth]{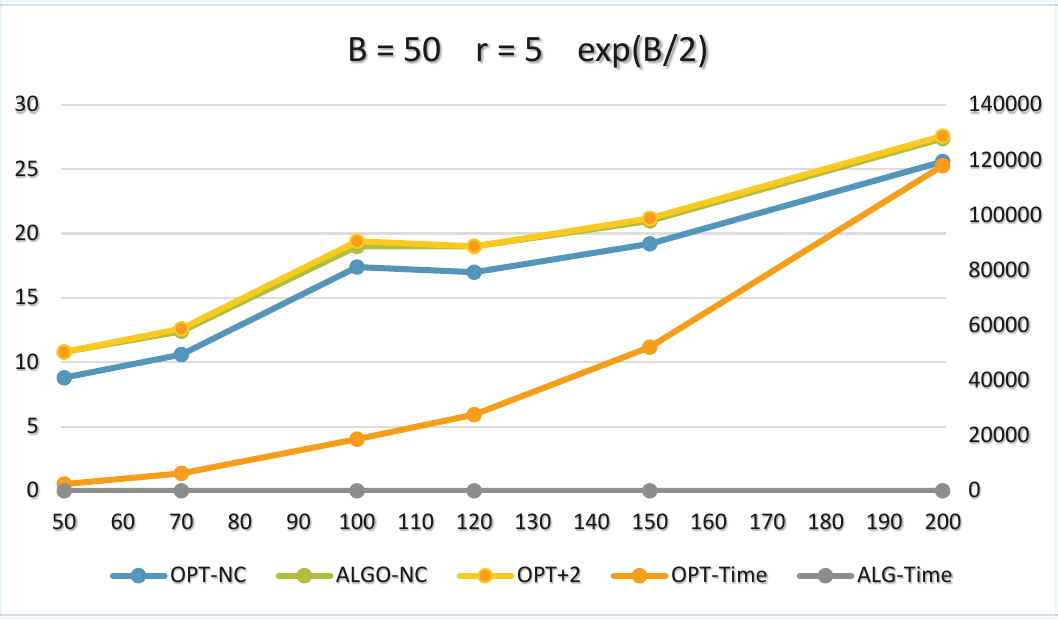}
\end{minipage}
\caption{Performance Evaluation for DDP-NC.}
\label{fig:ddp-nc}
\end{figure}

% \begin{figure}[ht]
% \centering
% \includegraphics[width=1\columnwidth]{sc-uniform-3_cropped.pdf}

% \includegraphics[width=\columnwidth]{sc-uniform-5_cropped.pdf}

% \includegraphics[width=\columnwidth]{sc-exp-3_cropped.pdf}

% \includegraphics[width=\columnwidth]{sc-exp-5_cropped.pdf}
% \caption{Performance Evaluation for DDP-SC.}
% \label{fig:ddp-sc}
% \end{figure}
Fig.~\ref{fig:ddp-sc} presents a performance comparison between our proposed algorithms  \textsc{AlgoFor-DDP-NC}, \textsc{Mod-Alg-DDP-SC} and the optimal solution for DDP-SC.
We here also consider four settings with $r \in \{3,5\}$, where the interval lengths follow either a uniform or an exponential distribution. 
For DDP-SC, we consider instances with $n \le 80$, as for $n \ge 90$, the number of constraints exceeds $10^5$, making it difficult to compute the corresponding optimal solution.
Here, we set the battery budget to 50. 
Across all instances, we observe that the solution returned by the modified algorithm \textsc{Mod-Alg-DDP-SC} is significantly less than that of the algorithm \textsc{AlgoFor-DDP-NC}.
Moreover, we find that the solution returned by the algorithm \textsc{AlgoFor-DDP-NC} is nearly thrice the optimal, and never crosses $4\,OPT\text{-}SC$. The modified algorithm \textsc{Mod-Alg-DDP-SC} is nearly twice the optimal, but never crosses $3\,OPT\text{-}SC$. These bounds align with the theoretical bound proven in Theorem \ref{thm:alg-ddp-sc}, \ref{thm:final-ddp-sc-mod}.
Additionally, the optimal solution for DDP-SC runs in exponential time, whereas both of our proposed algorithms take nearly 1 millisecond on the chosen instance. 
\begin{figure}[htbp]
\centering
\begin{minipage}{0.48\columnwidth}
    \centering
    \includegraphics[width=\linewidth]{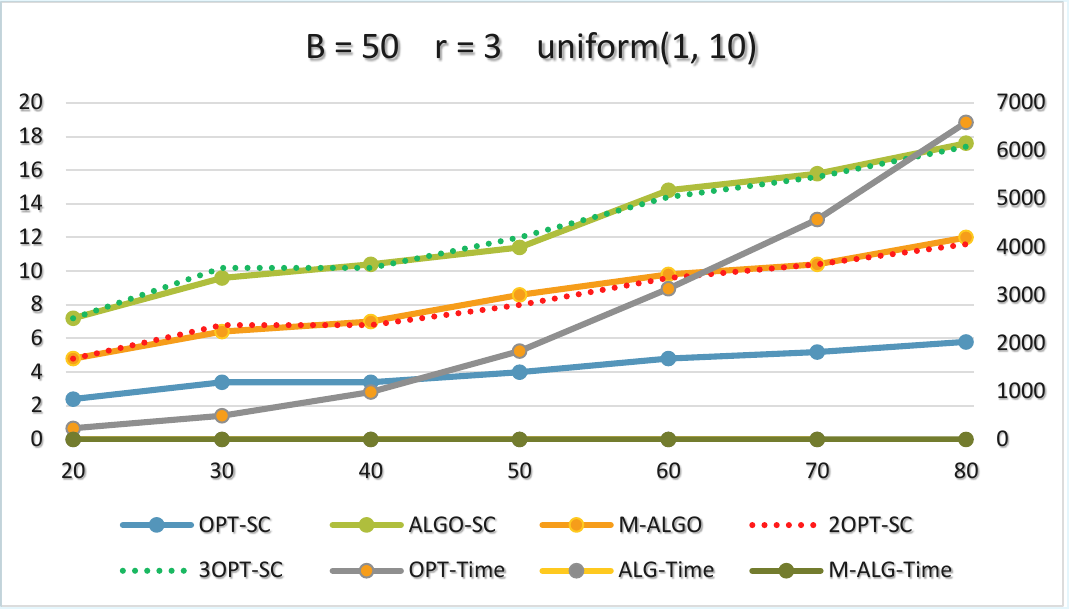}
\end{minipage}
\hfill
\begin{minipage}{0.48\columnwidth}
    \centering
    \includegraphics[width=\linewidth]{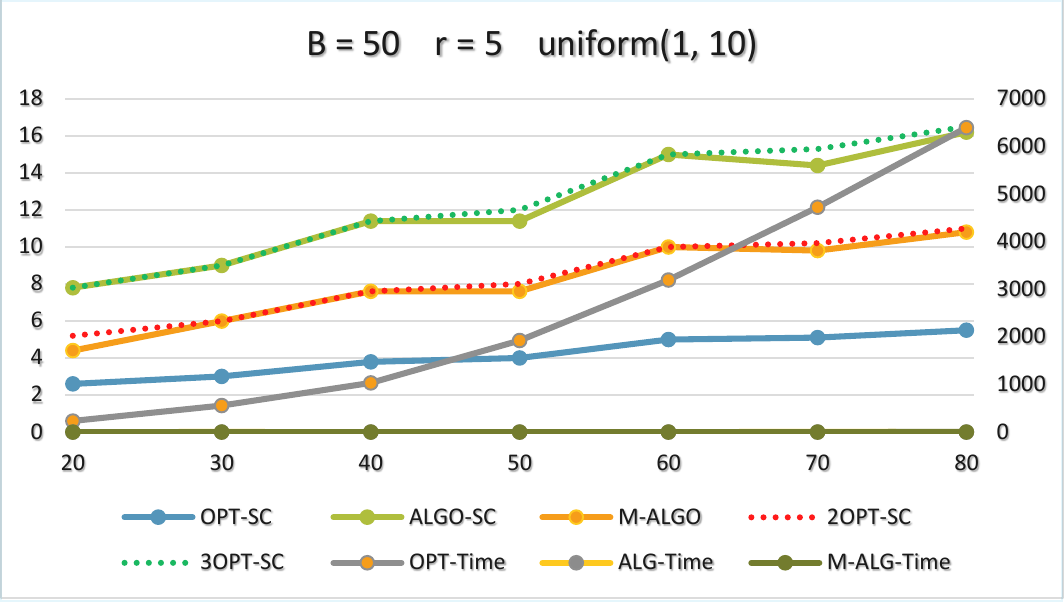}
\end{minipage}
\end{figure}
\begin{figure}[htbp]
\centering
\begin{minipage}{0.48\columnwidth}
    \centering
    \includegraphics[width=\linewidth]{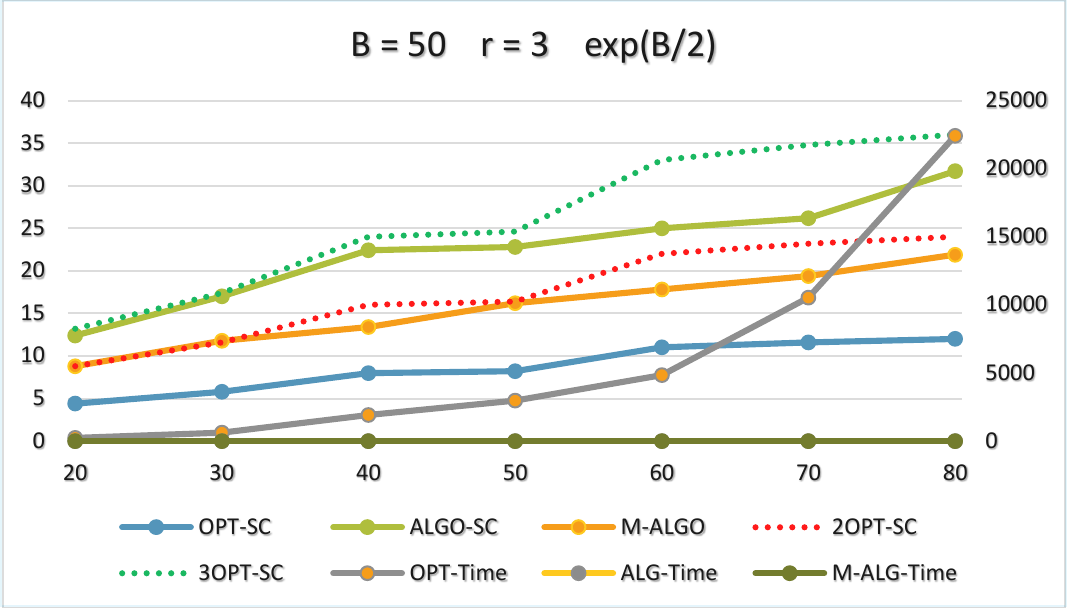}
\end{minipage}
\hfill
\begin{minipage}{0.48\columnwidth}
    \centering
    \includegraphics[width=\linewidth]{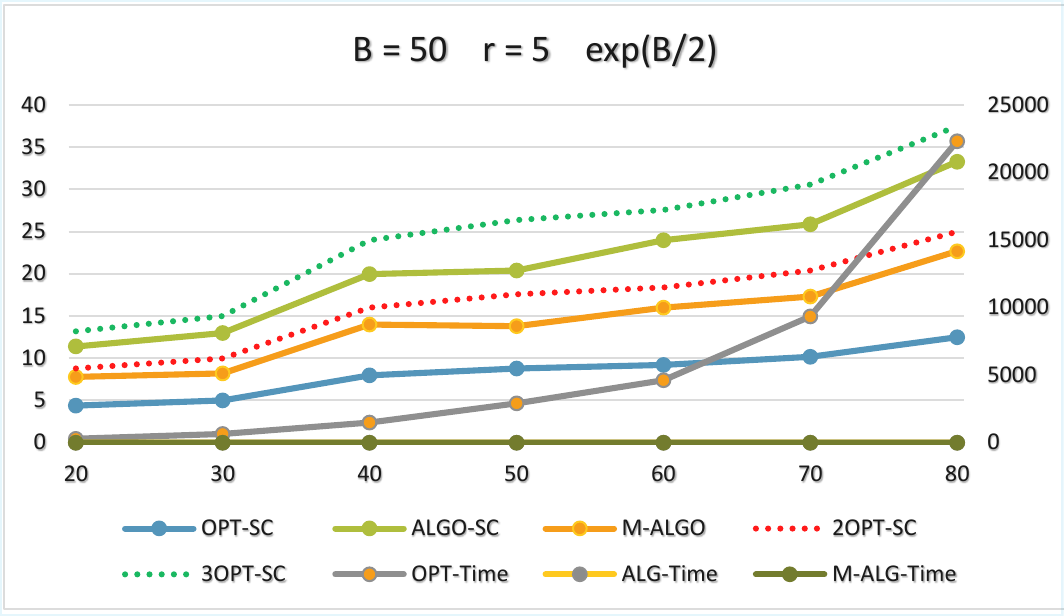}
\end{minipage}
\caption{Performance Evaluation for DDP-SC.}
\label{fig:ddp-sc}
\end{figure}

\section{Conclusion}
\label{section:concl}
In this paper, we studied the Drone-delivery Packing Problem. We discussed about three variants of the problem based on the conflicting characteristics of the delivery intervals and the presence of battery service stations. For each variant, we propose a distinct approximation algorithm and empirically compare its performance with the optimal solution.
We design our approximation algorithm to work for both types of battery service stations: swapping and recharging. Moreover, our algorithms for DDP-NC and DDP-SC can adapt any approximation algorithm for bin-packing and DDP-NS, respectively.
Finding better constant-factor approximation algorithms and asymptotic polynomial-time approximation schemes (PTAS) for all variants will be considered for future research. Additionally, it will be more challenging to consider a further realistic situation in an online environment.

%\section*{Acknowledgments}
% This should be a simple paragraph before the References to thank those individuals and institutions who have supported your work on this article.

\bibliographystyle{IEEEtran}
\bibliography{arxiv}

% \begin{IEEEbiographynophoto}{Saswata Jana}
% % Use $\backslash${\tt{begin\{IEEEbiographynophoto\}}} and the author name as the argument followed by the biography text.
% is currently a Prime Minister Research Fellow (PMRF) at the Department of
% Mathematics, Indian Institute of Technology Guwahati, India, since 2021. He pursued his M.Sc. in
% Mathematics (Gold Medalist) from the Indian Institute of Technology Patna, India, in 2021. He obtained
% his B.Sc. with Honours in Mathematics from the University of Calcutta, India, in 2019. He was a
% recipient of the Aryabhatta Gold Medal (2021) and Institute Silver Medal (2021) from IIT Patna, India.
% His present research interests include Approximation Algorithms and Distributed Algorithms.
% \end{IEEEbiographynophoto}
% \begin{IEEEbiographynophoto}{Partha Sarathi Mandal}
    
% \end{IEEEbiographynophoto}
%\appendices
\section{Appendix}
\label{sec:appendix}

\subsection{ILP}
\label{subsec:ILP}
Here, we present an integer linear program for the DDP-SC, assuming the battery stations are the swapping stations. 
Recall that $\mathcal{N} = \{1,2,\ldots,n\}$ denotes the set of deliveries and $\mathcal{R} = \{1,2,\ldots,r\}$ denotes the set of battery stations. Since each delivery and battery station corresponds to an interval, we combine these two sets and define $\mathcal{N}^+ = \mathcal{N} \cup \mathcal{R}$. For convenience, we relabel the elements of $\mathcal{N}^+$ as ${1,2,\ldots,n+r}$ in non-decreasing order of the left endpoints of their corresponding intervals.
Furthermore, we assume that the first index $1 \in \mathcal{N}^+$ corresponds to a delivery; otherwise, if it corresponds to a battery station, it can be ignored as it behaves equivalently to the warehouse.
Since, for all $j \in \mathcal{N}$, $cost(j) \le B$, $n$ is the upper bound for $OPT_{SC}$.
Let $\mathcal{M} = \{1, 2, \cdots, n\}$ be the set of drones available in the warehouse. Our objective is to use the minimum number of drones.

Let $x_{ij}$ be the binary variable, which is $1$ if the drone $i \in \mathcal{M}$ is assigned for the delivery (or the station) $j \in \mathcal{N}^+$; otherwise $0$.
$y_i$ is the binary variable that is $1$ if the drone $i \in \mathcal{M}$ is used.
$u_{ij}$ represents the battery level of the drone $i$ after processing the interval $j \in \mathcal{N}^+$.
The variable $z_{jk}$ is 1 if the intervals corresponding to the indices $j, k \in \mathcal{N}^+$ are in conflict; otherwise $0$. The variable $M >> B$ is a very large constant.

\begin{equation}
    \min \sum\limits_{i \in \mathcal{M}} y_i
    \label{eq:obj}
\end{equation}
\begin{equation}
    \text{subject to~~~} x_{ij} \le y_i~~~~~ \forall i \in \mathcal{M},~ \forall j \in \mathcal{N}^+
    \label{eq:C1}
\end{equation}
\begin{equation}
    \sum\limits_{i \in \mathcal{M}} x_{ij} = 1, ~~~~\forall j \in \mathcal{N}^+ \cap \mathcal{N}
    \label{eq:C2}
\end{equation}
\begin{equation}
    x_{ij} + x_{ik} \le 1, \forall i \in \mathcal{M}, \forall j, k \in \mathcal{N}^+ \text{~with~} z_{jk} = 1
    \label{eq:C3}
\end{equation}
\begin{equation}
    0 \le u_{ij} \le B ~~\forall i \in \mathcal{M}, ~~\forall j \in \mathcal{N}^+
    \label{eq:C4}
\end{equation}
\begin{equation}
    u_{i1} = B - cost(1) \cdot x_{i1}, ~~\forall i \in \mathcal{M}
    \label{eq:C5}
\end{equation}
\begin{equation}
    u_{ij} = u_{ij-1} - cost(j) \cdot x_{ij}, \forall i \in \mathcal{M}, \forall j \in \mathcal{N}^+ \cap \mathcal{N} \setminus \{1\}
    \label{eq:C6}
\end{equation}
\begin{equation}
    u_{ij} \ge B \cdot x_{ij}, ~\forall i \in \mathcal{M}, \forall j \in \mathcal{N}^+ \cap \mathcal{R}
    \label{eq:C7}
\end{equation}
\begin{equation}
    u_{ij} \le u_{ij-1} + M\cdot x_{ij}, ~\forall i \in \mathcal{M}, \forall j \in \mathcal{N}^+ \cap \mathcal{R}
    \label{eq:C8}
\end{equation}
\begin{equation}
    u_{ij} \ge u_{ij-1} - M\cdot x_{ij}, ~\forall i \in \mathcal{M}, \forall j \in \mathcal{N}^+ \cap \mathcal{R}
    \label{eq:C9}
\end{equation}
\begin{equation}
    cost(1) \cdot x_{i1} \le B, ~~\forall i \in \mathcal{M}
    \label{eq:C10}
\end{equation}
\begin{equation}
    cost(j) \cdot x_{ij} \le u_{ij-1}, \forall i \in \mathcal{M}, \forall j \in \mathcal{N}^+ \cap \mathcal{N} \setminus \{1\}
    \label{eq:C11}
\end{equation}
\begin{equation}
    x_{ij},~ y_i \in \{0,1\}, ~~\forall i \in \mathcal{M}, ~\forall j \in \mathcal{N}^+
    \label{eq:C12}
\end{equation}

Eq. \ref{eq:obj} is the objective function. Eq. \ref{eq:C1} says $x_{ij}$ is one only if the drone $i$ is used. Eq. \ref{eq:C2} tells that each delivery must be completed by exactly one drone.
Eq. \ref{eq:C3} confirms that if two intervals are in conflict, then a drone $i$ can not accomplish both of them.
Eq. \ref{eq:C4} tell that the battery level of a drone at any time must not exceed $B$ and does not go down $0$.
Eqs. \ref{eq:C5}~-~\ref{eq:C6} says that if the delivery $j$ is completed by the drone $i$, then after the delivery the battery level of the drone $i$ decreases by $cost(j)$ from its previous level.
However, if the delivery $j$ is not completed by drone $i$, the battery level remains the same. 
Eqs. \ref{eq:C7}~-~\ref{eq:C9} depicts that if the drone decides to swap its battery at the station $j \in \mathcal{N}^+ \cap \mathcal{R}$, then the post battery level of the drone $i$ after the swapping is $B$.
If it does not, its battery level remains the same as before.
Eqs. \ref{eq:C10}~-~\ref{eq:C11} tells that if the drone $i$ is assigned for the delivery $j$, then the current battery level of the drone $i$ must be at least $cost(j)$.

For DDP-NS, we remove the variable $u_{ij}$ and replace Eqs. \ref{eq:C4}~-~\ref{eq:C11} by $\sum_{j \in \mathcal{N}} cost(j)\cdot x_{ij} \le B$, $\forall i \in \mathcal{M}$.
For DDP-NC, the number of constraints in Eq. \ref{eq:C3} will be significantly lower, as there is no conflict between the two delivery intervals.

\subsection{Illustration of \textsc{AlgoFor-DDP-NC} with an example}
\label{sec:ex-ddp-nc}
Consider an instance of DDP-NC, as shown in Fig. \ref{fig:placeholder}, with 3 battery stations and 23 deliveries. Some intervals are represented by the black dots, and the size of the dot represents the cost of the corresponding interval. For simplicity, in the figure, we label the delivery interval $I_j$ by its index $(j)$ itself.
After Step-1, we have $\cI_1 = \{I_1, I_2, \cdots, I_6\}$; $\cI_2 =\{I_7, I_8, \cdots, I_{13}\}; \cI_3=\{I_{14}, I_{15}, \cdots, I_{19}\};$ and $\cI_4 =\{I_{20}, I_{21}, I_{22}, I_{23}\}$. 

Let $B=10$; $cost(I_1) = 3; cost(I_2) = 5; cost(I_3) = 9;cost(I_4) = 2;cost(I_5) = 3;$ and $cost(I_6) = 3$.
%Then, $\{I_3, I_2, I_1, I_5, I_6, I_4\}$ is the sorted interval set $\cI_1$.
If we apply FFD to $\cI_1$, it divides $\cI_1$ into 3 blocks, say $\{I_3\};~\{I_2,I_1,I_4\};~ \{I_5,I_6\}$. Similarly, let FFD divides $\mathcal{I}_2$ into 3 blocks: $\{I_7,I_{10}, I_{13}\};~ \{I_8,I_{11}\}; ~ \{I_9,I_{12}\};$~ $\mathcal{I}_3$ into another three blocks: $\{I_{14},I_{16},I_{18}\}; ~ \{I_{15},I_{19}\}; ~ \{I_{17}\};$ and $~\mathcal{I}_4$ into two blocks: $\{I_{20}, I_{22}\}; ~ \{I_{21}, I_{23}\}$.

We have, $I_1^{last} = I_6;~ I_2^{first} = I_7;~ I_2^{last} = I_{13};~ I_3^{first} = I_{14};~ I_3^{last} = I_{19}$;~ and $I_4^{first} = I_{20}$. Therefore, $S_1^{last} = \{I_5, I_6\};~ S_2^{first} = \{I_7, I_{10}, I_{13}\} = S_2^{last};~ S_3^{first} = \{I_{14}, I_{16}, I_{18}\};~ S_3^{last} = \{I_{15}, I_{19}\}$;~ and $S_4^{first} = \{I_{20}, I_{22}\}$.
Since, $m_1 = m_2 = m_ 3= 3$ and $m_4 =2$, we have $m_{max} = 5$. Let the $5$ used drones are $\{D_1, D_2, \cdots, D_5\}$. 
For $\cI_1$, we assign the block $\{I_3\}$ to $D_1$; the block $\{I_2, I_1, I_4\}$ to $D_2$; and the block $\{I_5, I_6\}$ to $D_3$. So, $D_1^{last} = D_3$. Then, we recharge/swap the drones $D_1$ and $D_2$ for the entire waiting interval $I_1^c$. Both drones will be at full capacity before processing the interval $I_8$.

For $\cI_2$, we assign the block $S_2^{first}$ to one of the drones except $D_0^{last} ( = \phi)$ and the drones assigned to $\cI_1$. So, we assign the block $\{I_7, I_{10}, I_{13}\}$ to $D_4$. Then the block $\{I_8, I_{11}\}$ assign to $D_1$ and the block $\{I_9, I_{12}\}$ assign to $D_2$. So, $D_2^{last} = D_4$. Thereafter, we recharge/swap the drones $D_1$, $D_2$ and $D_3$ for the entire waiting interval $I_2^c$. All of them will be at full capacity before processing the interval $I_{15}$.

For $\cI_3$, we assign the block $S_3^{first}$ to one of the drones except $D_1^{last} ( = D_3)$ and the drones assigned to $\cI_2$. So, we assign the block $\{I_{14}, I_{16}, I_{18}\}$ to $D_5$. Then the block $\{I_{15}, I_{19}\}$ is assigned to $D_1$ and the block $\{I_{17}\}$ is assigned to $D_2$. So, $D_3^{last} = D_1$. Thereafter, we recharge/swap the drones $D_2$, $D_4$ and $D_5$ for the entire waiting interval $I_3^c$. All of them will be at full capacity before processing the interval $I_{21}$.
Finally for $\cI_4$, we assign the block $S_4^{first}$ to one of the drones except $D_2^{last} ( = D_4)$ and the drones assigned to $\cI_3$. So, we assign the block $\{I_{20}, I_{22}\}$ to $D_3$. Then the block $\{I_{21}, I_{23}\}$ is assigned to $D_2$, and the algorithm terminates.
\begin{figure}[t]
    \centering
    \includegraphics[width=0.7\linewidth]{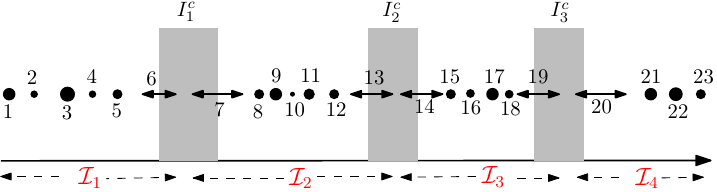}
    \caption{An example of DDP-NC with 3 waiting intervals and 23 delivery intervals.}
    \label{fig:placeholder}
\end{figure}
\subsection{Illustration of \textsc{Mod-Alg-DDP-SC}}
\label{sec:app-mod-ddp-sc}
\begin{figure}[ht]
    \centering
    \includegraphics[width=0.7\linewidth]{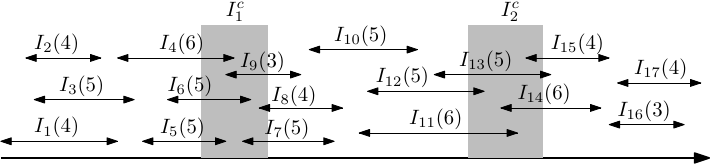}
    \caption{An example of DDP-SC with two waiting time intervals and 17 delivery intervals.}
    \label{fig:mod-ddp-sc}
\end{figure}

Consider an instance of DDP-SC, as shown in Fig. \ref{fig:mod-ddp-sc}, with 2 battery stations and 17 deliveries. $I_j(c)$ represents the delivery interval to the delivery $j$ with $cost(I_j) = c$. We assume $B = 10$. From the figure, it is easy to observe that $\omega$ is $3$. 
After the division of the intervals as of Step$-2^+$, $\widetilde{\cI}_1 = \{I_1, I_2, \cdots, I_9\}; ~\widetilde{\cI}_2 =\{I_{10}, I_{11}. \cdots, I_{15}\}$, and $\widetilde{\cI}_3 = \{I_{16}, I_{17}\}$.

We here illustrate the algorithm for $\ell =1$. The graph $G_1$ is shown in Fig. \ref{fig:G_1}. The vertex $j$ represents the vertex corresponding to the interval $I_j$. We have $V_{1}^{left} =\{4,5,6\}$ and $V_{1}^{right} = \{7,8,9\}$. Then we construct the bipartite graph $G_1^b(V_1^{left}\cup V_1^{right} , E_1)$, as depicted in Fig. \ref{fig:G_1_bipartite}. Since $I_4$ intersects with $I_9$, there is no edge between the vertices $4$ and $9$ in $G_1^b$. Also, since $cost(I_4) + cost(I_7) > B = 10$, there is no edge between the vertices $4$ and $7$ in $G_1^b$. However, being the intervals $I_4$ and $I_8$ compatible and $cost(I_4) + cost(I_8) \leq B$, there is an edge between $4$ and $8$ in $E_1$. Similarly, we add the other edges, as shown in Fig. \ref{fig:G_1_bipartite}.

\begin{figure}[h]
\centering
\linenumbers
\begin{minipage}{0.32\textwidth}
    \centering
    \includegraphics[width=.6\linewidth]{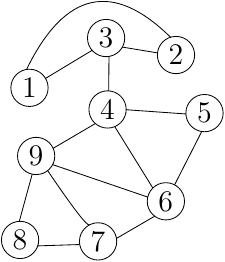}
    \caption{Induced subgraph $G_1$ of $G$ from  vertices of $\widetilde{\cI}_1$.
   % The vertex $j$ represents the interval $I_j$, where $1 \leq i \leq 9$.
    }
    \label{fig:G_1}
\end{minipage}
\hfill
\begin{minipage}{0.32\textwidth}
    \centering
    \includegraphics[width=.55\linewidth]{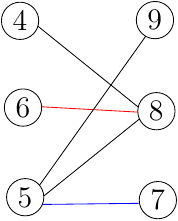}
    \caption{Bipartite graph $G_1^b$ with matching colored edges.
   % , where $V_{1}^{left} =\{4,5,6\}$ and $V_{1}^{right} = \{7,8,9\}$.
    }
    \label{fig:G_1_bipartite}
\end{minipage}
\hfill
\begin{minipage}{0.32\textwidth}
    \centering
    \includegraphics[width=.55\linewidth]{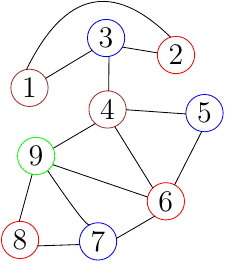}
    \caption{Coloring the vertices of $G_1^b$ from the matching.
   % , where $V_{1}^{left} =\{4,5,6\}$ and $V_{1}^{right} = \{7,8,9\}$.
    }
    \label{fig:G_1_color}
\end{minipage}
\end{figure}

The size of the maximum matching of $G_1^b$ is 2, let $(5,7)$ and $(6,8)$ be the matching edges. So, $x_1 = 2$, and $z_1 = 4$. We now color the vertices of $G_1$ with $\max\{\omega, z_1\}$ = 4 colors. The vertices $5$ and $7$ get the same color, say ``blue". Then the vertices $6$ and $8$ get the color same color, say ``red". Then, the unmatched vertices $4$ and $9$ get the different colors, say ``brown"
and ``green", respectively. Thereafter, the vertex $3$ is colored with ``blue", then the vertex $1$ is colored with ``brown", and finally the vertex $2$ with ``red". 

Applying Step-$6^+$, $\widetilde{\cI}_1$ is divided into $6$ blocks. For the ``blue" colored interval, intervals $I_5$ and $I_7$ are grouped into one block, then $I_3$ into another one. For the ``red" colored intervals, intervals $I_6$ and $I_8$ are grouped into one block, then $I_2$ into another one. Similarly, the interval $I_9$ is packed into one block, and the intervals $I_1$ and $I_4$ are packed into one block. 

Similarly, we will get $z_2 = 4$ and $\widetilde{\cI}_2$ is divided into four blocks: $\{I_{12}, I_{15}\}; ~\{I_{10}, I_{13}\};~ \{I_{11}\};~ \{I_{14}\}$, and $\widetilde{\cI}_3$ into two blocks: $\{I_{16}\}; ~\{I_{17}\}$.
Therefore, $\widetilde{m}_1 = 6;~ \widetilde{m}_2 = 4$; and $\widetilde{m}_3 = 2$, implies $\widetilde{m}_{max} = 6$. Moreover, $z_1 =4; ~z_2 = 4$, implies $z_{max} = 4$. Hence, the algorithm used $(6+4) = 10$ drones.
%\newpage
\subsection{List of Variables}
% \section{Appendix}
% \label{sec:appendix}
\begin{table}[h]
    \centering
    \scriptsize{
    \begin{tabular}{|c|c|c|c|}
        \hline
         $\mathcal{N}$& Set of deliveries & $B$  & Battery budget of the drone\\ \hline
      %  $\delta_j^L$ & Launching location of the delivery $j$  &  $\delta_j^R$ & Rendezvous location of the delivery $j$ \\ \hline
        $t_j^L$ & Launching time of the delivery $j$  &  $t_j^R$ & Rendezvous time of the delivery $j$ \\ \hline
        $I_j$ & Delivery time interval of the delivery $j$  & $\cI$ & Set of delivery time intervals\\ \hline
        $cost(I_j)$ & Cost for the delivery $j$  & $n$& Total number of deliveries \\ \hline
        $S_i^t$ & Set of deliveries assigned to drone $i$ at time $t$  & $rem_i^t$ & Remaining battery of drone $i$ at time $t$ \\ \hline
        $\mathcal{R}$& Set of battery service stations & $r$ & Total number of battery stations\\ \hline
         $t_{\ell}^A$& Arrival time of the truck at the station $\ell$  & $t_{\ell}^D$& Departure time of the truck from the station $\ell$ \\ \hline
         $I_{\ell}^c$& Waiting time interval at station $\ell$ & $\cI^c$ & Set of waiting time intervals\\ \hline
         %$n$& Total number of deliveries& $r$ & Total number of battery stations\\ \hline
         $OPT_{NS}$& Minimum no. of drones needed for DDP-NS & $OPT_{NC}$ & Minimum no. of drones needed for DDP-NC\\ \hline
         $OPT_{SC}$ & Minimum no. of drones needed for DDP-SC & $\omega$ & Maximum no. of pairwise conflicting intervals\\ \hline
         $G$& Interval graph to interval set $\cI$ & $\mathcal{J}_k$ & Set of intervals with color $k$\\ \hline
         $n_e$&Total no. of conflicts among the intervals in $\cI$ & $m_k$ &Number of drones used for $\mathcal{J}_k$\\ \hline
         $\epsilon_k$& $\frac{1}{B}\max \limits_{I_j \in \J_k} ~cost(I_j)$ & $\epsilon'_k$& $\min\{\frac{1}{2}, \epsilon_k\}$.\\ \hline
         $\epsilon_{min}$ & $ \frac{1}{B} \min \limits_{I_j \in \mathcal{J}_k}~cost(I_j)$& $\epsilon_{max}$& $ \max \limits_{1 \leq k \leq \omega} \epsilon'_k$\\ \hline
         $I_{\ell}^{first}$&Interval of $\cI_{\ell}$ containing $t_{\ell-1}^D$ & $I_{\ell}^{last}$ & Interval of $\cI_{\ell}$ containing $t_{\ell}^A$\\\hline
         $S_{\ell}^{first}$ & The block assigned for $I_{\ell}^{first}$ & $S_{\ell}^{last}$ & The block assigned for $I_{\ell}^{last}$\\\hline
         $D_{\ell}^{first}$ & The drone assigned for $I_{\ell}^{first}$ & $D_{\ell}^{last}$ & The drone assigned for $I_{\ell}^{last}$\\\hline
         $OPT_{\ell}$ & Min. no. of blocks needed for \texttt{Partition}$(\cI_{\ell})$ & $S_{\ell}^i$ & $i$-th block of $\cI_{\ell}$ returned by the algorithm\\ \hline
         $m_{\ell}$ & No. of blocks of $\cI_{\ell}$ returned by the algorithm & $m_{max}$ &  $\max\limits_{1 \le \ell \le r+1} m_{\ell}$\\ \hline
         $\cI_{\ell}^{first}$&Set of intervals of $\cI_{\ell}$ containing $t_{\ell-1}^D$ & $\cI_{\ell}^{last}$ & Set of intervals of $\cI_{\ell}$ containing $t_{\ell}^A$\\\hline
         $\mathcal{S}_{\ell}^{first}$ & Set of blocks assigned for $I_{\ell}^{first}$ & $\mathcal{S}_{\ell}^{last}$ & Set of blocks assigned for $I_{\ell}^{last}$\\\hline
         $\mathcal{D}_{\ell}^{first}$ & Set of drones assigned for $I_{\ell}^{first}$ & $\mathcal{D}_{\ell}^{last}$ & Set of drones assigned for $I_{\ell}^{last}$\\\hline
         $\widetilde{\cI}_{\ell}$& $\ell$-th modified interval set & $G_{\ell}$& Subgraph of $G$ induced by the vertices of $\widetilde{\cI}_{\ell}$\\ \hline
         $G_{\ell}^b$& Bipartite subgraph of the complement of $G_{\ell}$& $E_{\ell}^b$ &Set of edges in $G_{\ell}^b$\\ \hline
         $\cI_{\ell}^{left}$& Set of intervals in $\widetilde{\cI}_{\ell}$ containing $t_{\ell}^A$ & $\cI_{\ell}^{right}$& Set of intervals in $\widetilde{\cI}_{\ell}$ containing $t_{\ell}^D$\\ \hline
         $V_{\ell}^{left}$& Set of vertices cor. to the intervals in $\cI_{\ell}^{left}$ & $V_{\ell}^{right}$& Set of vertices cor. to the intervals in $\cI_{\ell}^{right}$\\ \hline
         $\mathcal{M}_{\ell}$& Maximum matching of $G_{\ell}^b$& $x_{\ell}$& Size of the matching of $\mathcal{M}_{\ell}$\\ \hline
         $z_{\ell}$& Sum of $x_{\ell}$ and unmatched vertices of $G_{\ell}^b$& $\widetilde{m}_{\ell}$& No. of blocks of $\widetilde{\cI}_{\ell}$ returned by algorithm \\ \hline
         $\cI_{\ell}^{ext}$ & Set of intervals in $\cI_{\ell}^{left}$ and $\cI_{\ell}^{right}$ & $\mathcal{D}_{\ell}^{ext}$ & Set of drones assigned for $\cI_{\ell}^{ext}$\\ \hline
         $z_{max}$ &$\max\limits_{1 \le \ell \le r+1} z_{\ell}$ & $\widetilde{m}_{max}$ &$ \max\limits_{1 \le \ell \le r+1} \widetilde{m}_{\ell}$\\ \hline
         \end{tabular}}
         
    \caption{List of variables used in this paper}
    \label{tab:notation}
\end{table}
%\label{sec:variable-list}
\end{document}